\documentclass[a4paper,10pt]{article}
\usepackage{amsmath}
\usepackage{amsthm}
\usepackage{amsfonts}
\usepackage{amssymb}
\usepackage{graphicx}
\usepackage{color}
\usepackage{fullpage}
\usepackage{multirow}
\usepackage{subfig}
\usepackage{hyperref}
\usepackage{booktabs}
\usepackage{xspace}
\usepackage[numbers,sort&compress]{natbib}
\usepackage{epstopdf}                       % Handles eps figs when using pdflatex
\usepackage{pstool}
\usepackage{psfrag}
\newcommand{\munit}[1]{[\mathrm{#1}]}

\newtheorem{lemma}{Lemma}
\newtheorem{theorem}{Theorem}

\newtheorem{assumption}{Assumption} 
\newtheorem{condition}{Condition}
\theoremstyle{remark}  \newtheorem{remark}{Remark}
\newcommand{\bsym}[1]{\boldsymbol{#1}}

\newcommand{\mc}[1]{\mathcal{#1}}

\allowdisplaybreaks

\pdfoutput = 1

\title{Observer Based Path Following for Underactuated Marine Vessels in the Presence of Ocean Currents: A Global Approach \\
With proofs 
}    
\date{}
\author{ D.J.W. Belleter \thanks{Contact information: {\tt dennisbelleter@gmail.com}} \footnotemark[5] \and M. Maghenem \thanks{Department of electrical and computer engineering, University of California, Santa Cruz. {\tt mmaghene@ucsc.edu}} \and C. Paliotta \thanks{Contact information: {\tt claudiopaliotta@ieee.org}} \footnotemark[5] \and K.Y. Pettersen \thanks{ Centre for Autonomous Marine Operations and Systems (NTNU AMOS), Department of Engineering Cybernetics, Norwegian University of Science and Technology, NO7491 Trondheim, Norway {\tt kristin.y.pettersen@ntnu.no}} \thanks{D.J.W. Belleter, C. Paliotta, and K.Y. Pettersen were supported by the Research Council of Norway through its Centers of Excellence funding scheme, project No. 223254 – AMOS.}}
\begin{document}
\maketitle 

\begin{abstract} 
In this paper a solution to the problem of following a curved path in the presence of a constant unknown ocean current disturbance is presented. We introduce a path variable that represents the curvilinear abscissa on the path which is used to propagate the path-tangential reference frame. The proposed dynamic update law of the path variable is non singular and the guidance law is designed such that the vessel can reject constant unknown ocean currents by using an ocean current observer. It is shown that the closed-loop system composed of the guidance law, controller and observer provides globally asymptotically stable and locally exponentially stable path following errors. The sway velocity dynamics is analyzed and, under adequate hypothesis on the path curvature, it is shown that the dynamics are well behaved and that the guidance law to exist. Simulations are presented to verify the theoretical findings.
\end{abstract}
%\begin{IEEEkeywords}
\section{Introduction}
This paper considers curved path-following for underactuated marine vessels. While the literature for straight-line path following of underactuated marine vessels is, by now, well established even in the presence of unknown disturbances, see e.g. \cite{aguiar2007trajectory, lapierre2007nonlinear, Borhaug2011, do2006global, oh2010path}, the literature for curved paths is much less rich and has some lacks. 

In the existing literature, we distinguish between local and global approaches to solving the path following problem for autonomous vehicles. In local approaches the path following problem is solved only when the the vehicle starts in a certain neighborhood of the path. This last fact is due to the singularity that appears in the dynamics of the path variable that propagate the path-tangential reference frame, since it is designed to keep the vessel on the normal of the frame at the path variable abscissa \cite{micaelli1993trajectory}. This method has the advantage of allowing faster convergence to the path compared to global approaches that solve the problem for all initial condition of the vehicle. The dynamics of the path variable, in this case, is used as a degree of freedom in the control design and is chosen to be nonsingular. 

The local approaches for the case of underactuated marine vehicles have been inspired by the seminal works in \cite{samson1992path} and \cite{micaelli1993trajectory} for the case of nonholonomic mobile robots. A first extension from mobile robots to the case of underactuated marine vessels appeared in \cite{encarnaccao2000bpath}, where the path parametrization from \cite{micaelli1993trajectory} is used to define the path-following problem, and a solution is presented using a nonlinear observer-based controller to incorporate the effects of an unknown but constant ocean current. Part of the closed-loop state is shown to be asymptotically stable and the zero dynamics is shown to be well behaved. Similar results for a $3$D underactuated marine vessels are presented in \cite{encarnacao20003d}. Another local result based on the path parametrization of \cite{samson1992path, micaelli1993trajectory} is provided in \cite{do2004state}, where only practical stability of the path-following errors is shown in the presence of environmental disturbances. Recently, in \cite{belleter2017observera}, we extended the latter result to guarantee exponential stability of the path following errors and to provide a complete analysis of the sway velocity.  

A global approach to solve the path following problem under general curved paths for underactuated marine vehicles is presented in \cite{lapierre2003nonlinear} and \cite{lapierre2007nonlinear} based on a result from the field of mobile robots in \cite{soetanto2003adaptive}. Solving this problem for this class of vehicles, offers the challenge of defining a controller which guarantees convergence of the vehicle to the path and at the same time gives boundedness of the sway velocity. In fact, when the vehicle moves along curved paths, the centrifugal effect causes a non-zero side velocity. In order to have a feasible motion of the vehicle, the controller has to guarantee a bounded sway velocity for curved paths. Note that for the case of straight line paths, the controller has to guarantee that the side velocity converges to zero since there is no centrifugal effect when following the path.

For the particular case of straight-line paths, a similar approach to \cite{lapierre2003nonlinear} is considered in \cite{borhaug2006path} in which a look-ahead based steering law is used to guide the vehicle to the path. Stability of the path-following errors is shown using cascaded systems theory, and the zero dynamics are analyzed and shown to be well behaved. In \cite{Borhaug2008} ocean currents are taken into account by adding integral action to the steering law. The work in \cite{Borhaug2008} is reformulated in \cite{caharija2012b} and \cite{caharija2012a} using relative velocities. Experimental results are obtained in \cite{caharija2016integral}. See also \cite{li2009design}, where a control design for straight-line paths has been validated by experiments. However in this work the sway dynamics are neglected in the control design procedure and stability analysis. Furthermore, for the particular case of circles and paths made of straight-line sections connecting way points, line-of-sight guidance laws are presented in \cite{breivik2004path} and \cite{Fossen2003}, respectively. In \cite{breivik2004path} the vessel is regulated to the tangent of its projection on the circle. An extension to the three dimensional case is given in \cite{breivik2005guidance} and \cite{breivik2005principles}. However, these works do not consider environmental disturbances.

For the case of general curved paths, we notice that a non-complete analysis of closed-loop dynamics has been provided in most of the existing literature. Specifically, the sway dynamics is not analysed and the existence and boundedness of the control input is not guaranteed \cite{lapierre2003nonlinear, lapierre2007nonlinear, moe2014path}. At the exception of \cite{paliotta2017trajectory}, where a different approach to the trajectory tracking and path following problems is proposed.
The approach presented in \cite{paliotta2017trajectory} is based on
a different choice of the output of the system, the so called hand position point. Then the authors apply an input-output linearizing controller in order to make the new output converge to the desired path. However, the controller presented in \cite{paliotta2017trajectory} is not applicable if the same output for the system as in \cite{lapierre2003nonlinear, lapierre2007nonlinear, moe2014path} is chosen. In this paper, we provide a rigoreous approach by keeping the traditional choice of the pivot point as output of the system and by modifying the control structure proposed in \cite{moe2014path}. These modifications allow us to guarantee the global asymptotic stability of the path following errors. Moreover, we derive sufficient conditions on the path curvature that allow us to prove the existence of the control law and the boundedness of the sway velocity. This is achieved by considering a global parametrization of the general curved path in order to solve the problem using a combination of an ocean current observer and a controller based on a \emph{line-of-sight-like} guidance. We consider underactuated vehicles, in particular, vehicles which do not have sway actuation. The guidance based controller proposed in this paper is said to be line-of-sight-like since it adopts a time-varying look-ahead distance depending on the path-following error. The time-varying look-ahead distance is modified compared to the one in \cite{moe2014path}, and it is shown that a new dependency on the path-following errors is crucial to prove boundedness of the sway velocity, which is the best behavior we can achieve for the zero dynamics in the case of general curved paths. It should be noted that these modifications are not merely an extension but are necessary conditions for the validity of stability results for the problem under consideration. To the best of our knowledge, such a result is unique in the literature of underactuated marine vehicles. Furthermore in \cite{moe2014path} the controller for the yaw rate dynamics used signals that dependent on the unknown ocean current. Therefore the controller could not be implemented. We specifically address this issue and derive a controller for the yaw rate dynamics that depends only on known signals.

The outline of the paper is as follows. In Section \ref{COG-sec:mdl} the vessel model is given. The path-following problem and the chosen path parametrization are introduced in Section \ref{COG-sec:pd}. Section \ref{COG-sec:ctrl} presents the ocean current observer that is used together with the guidance law and controllers. The closed-loop system is then formulated and analyzed in Section \ref{COG-sec:clsys}. A simulation case study is presented in Section \ref{COG-sec:case} and conclusions are given in Section \ref{COG-sec:cncl}.             
\section{Vehicle Model} \label{COG-sec:mdl}
In this section we introduce the vehicle's model given in \cite[p.152-157]{fossen2011handbook}. 
However, in the simulation section we explain why we use $u_{rd}=5m/s$ in the case studies.
This model can be used to describe an autonomous surface vessel or an autonomous underwater vehicle moving in the plane.  
The dynamics of the vehicle is:
\begin{subequations} \label{COG_eq:relVelMod}
\begin{align}
\dot{\eta} = & \; R(\psi) \nu_{r} + V \\
M \dot \nu_{r} + C(\nu_{r}) \nu_{r} + D \nu_{r} = & \; B f
\end{align}
\end{subequations}
where $\eta \triangleq [x,y,\psi]^T$ describes the position of the center of gravity and the orientation of the vehicle with respect to the inertial frame, $\nu_r \triangleq [u_r,v_r,r]^T$ contains the surge, the sway and the yaw velocities respectively, $M$ is the mass matrix, $C(\nu_{r})$ is the Coriolis matrix, $D$ is the damping matrix, $B$ is the thrust allocation matrix, and $f \triangleq [T_{u}, T_{r}]$ is the vector of control inputs composed by the surge thrust and the rudder angle inputs $T_{u}$ and $T_{r}$, respectively.
\\
For port-starboard symmetric vehicles, the matrices $(M, B, C, D)$ have the following structure
\begin{subequations}
\begin{align}
 M \triangleq &
\begin{bmatrix} 
m_{11}   &  0   & 0 \\
0       & m_{22} & m_{23} \\
0      &  m_{32} & m_{33}
\end{bmatrix}; B  \triangleq 
\begin{bmatrix}
b_{11}   & 0  \\
0           & b_{22}  \\
0           & b_{32}
\end{bmatrix}; \label{COG_eq:B}\\
C \triangleq & \begin{bmatrix}
0 & 0 & -m_{22} v_{r} - m_{23} r \\
0 & 0 & m_{11} u_{r} \\
m_{22}v_{r} + m_{23}r & -m_{11}u_{r} & 0
\end{bmatrix}; \\
D \triangleq & \begin{bmatrix}
d_{11}   & 0   & 0 \\
0           & d_{22} & d_{23} \\
0           & d_{32} & d_{33}
\end{bmatrix}.
\end{align}
\end{subequations}
It is worth noting that the model \eqref{COG_eq:relVelMod} is valid for low speed motion, for which the damping can be assumed to be linear. Specifically, at low speed the non-linear damping effects can be neglected \cite[p.152-157]{fossen2011handbook}. In the simulations in this paper, we use a model from \cite{caharija2014thesis} in which the damping is linear up to $\pm7~\munit{m/s}$.
Furthermore, since $f \in \mathbf{R}^{2}$, the vehicle is under-actuated in the work space $\mathbf{R}^{3}$. This latter fact implies that the vehicle is not directly actuated in the sway direction, that is, sideways. Moreover, given the structure of the matrix $B$ in \eqref{COG_eq:B}, it is easy to see that the control input in yaw $T_{r}$, indirectly affects the sway direction. However, according to \cite{fredriksen2004global}, for port-starboard symmetric vehicles, it is always possible to apply a change of coordinates such that the model \eqref{COG_eq:relVelMod}, can be expressed with respect to a coordinate frame positioned at the pivot point instead of at the center of gravity. The pivot point lies along the center line of the vehicle, ahead of the center of gravity and always exists for port-starboard vehicles \cite{fredriksen2004global}, and in this point the yaw control does not affect the sway motion. Hence, the dynamical model with the body-fixed frame positioned at the pivot point is the following: 
\begin{subequations} \label{COG-eq:dynsys}
\begin{align}  
\dot{x}&=u_r \cos(\psi) - v_r \sin(\psi) + V_x \label{COG-eq:xdot} \\
\dot{y}&=u_r \sin(\psi) + v_r \cos(\psi) + V_y \label{COG-eq:ydot} \\
\dot{\psi} &= r \label{COG-eq:phidot} \\
\dot{u}_r &= F_{u_r}(v_r,r)-\tfrac{d_{11}}{m_{11}}u_r + \tau_u \label{COG-eq:urdot}\\
\dot{v}_r &= X(u_r)r+Y(u_r)v_r \label{COG-eq:vrdot} \\
\dot{r} &= F_r(u_r,v_r,r)+\tau_r. \label{COG-eq:rdot}
\end{align}
\end{subequations}
%where $(x,y)$ and $\psi$ are, respectively, the position of the center of gravity and the orientation of the vehicle with respect to the inertial frame. The variables $u_r$, $v_r$ and $r$ are, respectively, the surge, the sway and the yaw velocities.  
where $\tau_u$ and $\tau_r$ are, respectively, the surge force and the yaw torque input. The functions $X(u_r)$, $Y(u_r)$, $F_u$, and $F_r$ are given by 
\begin{subequations}
\begin{align}
F_{u_r}(v_r,r) &\triangleq  \tfrac{1}{m_{11}}(m_{22}v_r+m_{23}r)r \\ 
X(u_r) &\triangleq \tfrac{m^2_{23}-m_{11}m_{33}}{m_{22}m_{33}-m^2_{23}}u_r + \tfrac{d_{33}m_{23}-d_{23}m_{33}}{m_{22}m_{33}-m^2_{23}}\\ \label{COG-eq:Xur}
Y(u_r) &\triangleq  \tfrac{(m_{22}-m_{11})m_{23}}{m_{22}m_{33}-m^2_{23}}u_r - \tfrac{d_{22}m_{33}-d_{32}m_{23}}{m_{22}m_{33}-m^2_{23}}\\ \label{COG-eq:Yur}
\begin{split}
F_r(\cdot)&\triangleq  \tfrac{m_{23}d_{22}-m_{22}(d_{32}+(m_{22}-m_{11})u_r)}{m_{22}m_{33}-m^2_{23}}v_r \\
&\quad +\tfrac{m_{23}(d_{23}+m_{11}u_r)-m_{22}(d_{33}+m_{23}u_r)}{m_{22}m_{33}-m^2_{23}}r.
\end{split}
\end{align}
\end{subequations}
Note that the functions $X(u_r)$ and $Y(u_r)$ are linear functions of the velocity. The kinematic variables are illustrated in Figure \ref{COG-fig:defin}. The ocean current satisfies the following assumption.
\begin{assumption} \label{COG-assum:current}
The ocean current is assumed to be a constant in time, uniform in space, and irrotational with respect to the inertial frame, i.e. $\bsym{V_c}\triangleq [V_x,V_y,0]^T$. Furthermore, there exists a constant $V_{\max}>0$ such that $\|\bsym{V_c}\|=\sqrt{V^2_{x}+V^2_{y}}\leq V_{\max}$.
\end{assumption}
We consider a range of values of the desired surge velocity $u_{rd}$ such that the following assumption holds.
\begin{assumption} \label{COG-assum:Yur}
It is assumed that $Y(u_r)$ satisfies 
\[ Y(u_r) \leq -Y_{\min}< 0, \,\forall u_r \in [-V_{\max},u_{rd}],\]
i.e. $Y(u_r)$ is negative for the range of desired velocities considered.
\end{assumption} 
Additionally we assume that the following assumption holds
\begin{assumption} \label{COG-assum:vel}
It is assumed that $u_{rd}(t)$ is $\mathcal{C}^1$ and satisfies $u_{rd}(t) > 2V_{\max}~\forall t$, i.e. the desired relative velocity of the vessel is larger than twice the maximum value of the ocean current.
\end{assumption} 
Assumption \ref{COG-assum:vel} assures that the vessel has enough propulsion power to overcome the ocean current affecting it. The factor two in Assumption \ref{COG-assum:vel} adds some extra conservativeness to bound the solutions of the ocean current observer, this is discussed further in Section \ref{COG-sec:ctrl}.
\begin{figure}[!htb]
\centering
\includegraphics[width=.5\columnwidth]{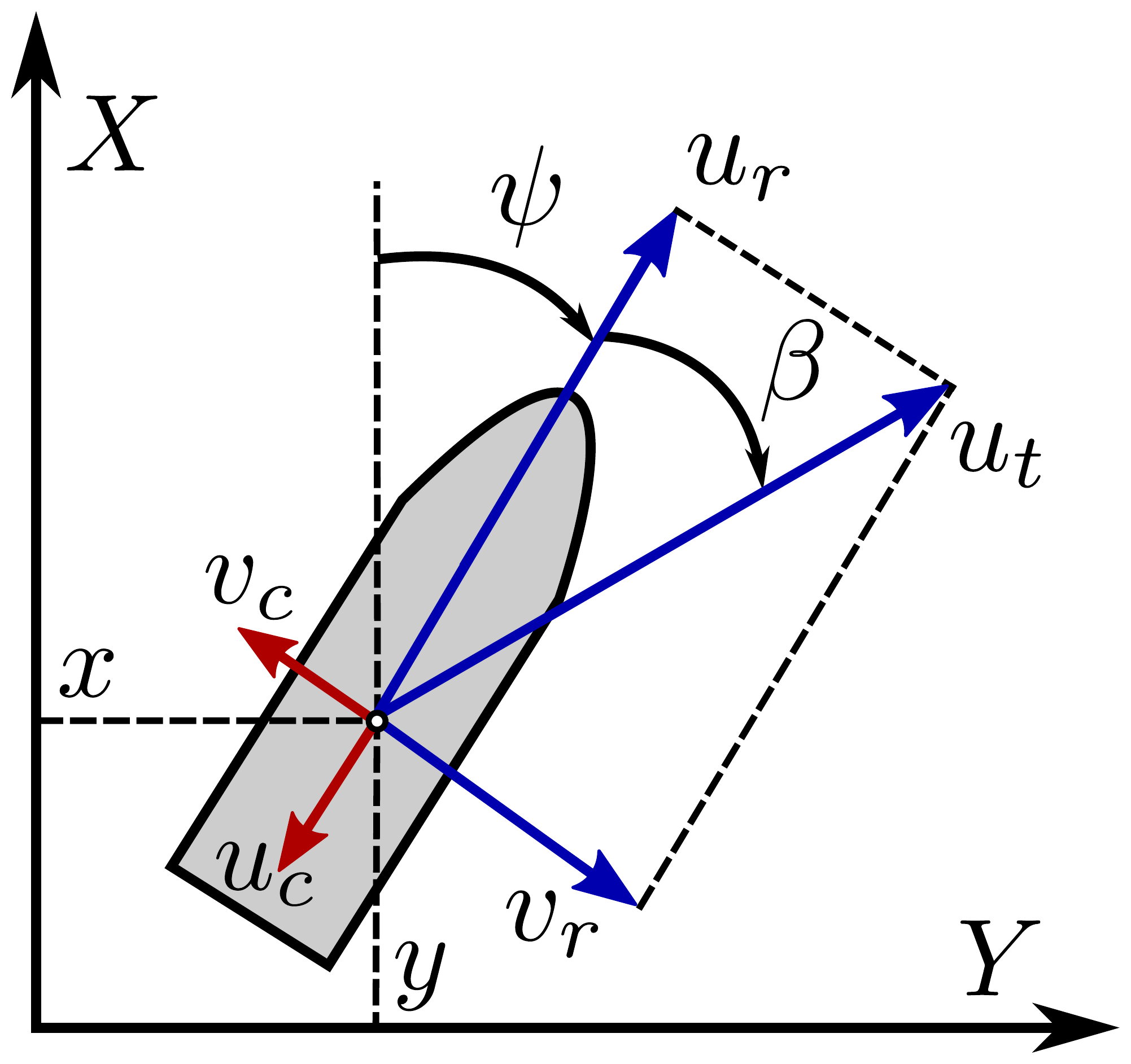}
\caption{Definition of the vehicle state variables.} \label{COG-fig:defin}
\end{figure}
\section{Problem definition} \label{COG-sec:pd}
The objective in global path-following problems is to make each trajectory of the controlled vehicle converge to a desired trajectory describing a smooth path $P$ regardless of the vehicle's initial location. For an underactuated vehicle, the path following task can be achieved by positioning the vehicle on the path with a total speed $u_t \triangleq \sqrt{u^2_r+v^2_r}$ (see Figure \ref{COG-fig:defin}) that is tangential to the path. However, this approach restricts the initial location of the vehicle to be on the path. To have a more general result in terms of both attractivity and invariance of the path, we introduce adequate path-following errors. The path-following errors correspond to the error between the vehicle and a point moving on the path. To do so, we parametrize the path $P$ using a path variable $\theta$. Moreover, for each point on the path, $(x_p(\theta), y_p(\theta)) \in P$, we introduce a path-tangential frame as illustrated in Figure \ref{COG-fig:path}. 
Hence, the path-following errors expressed in the tangential frame and denoted by $\bsym{p}_{b/p}\triangleq[x_{b/p},y_{b/p}]^T$ take the following form:
\begin{equation} \label{COG_eq:ye}
\begin{bmatrix} x_{b/p} \\ y_{b/p} \end{bmatrix} = \begin{bmatrix} \cos(\gamma_p(\theta)) & \sin(\gamma_p(\theta)) \\ -\sin(\gamma_p(\theta)) & \cos(\gamma_p(\theta)) \end{bmatrix} \begin{bmatrix} x - x_{p}(\theta) \\ y - y_{p}(\theta) \end{bmatrix}
\end{equation} 
where $\gamma(\theta)$ is the angle of the path with respect to the inertial $X$-axis. 
\\
The time derivative of the angle $\gamma(\theta)$ is given by $\dot{\gamma}(\theta) = \kappa(\theta)\dot{\theta}$ where $\kappa(\theta)$ is the curvature of $P$ at $\theta$. The path-following error is expressed by $x_{b/p}$ and $y_{b/p}$ which are the relative positions between the path frame and the body frame expressed along the axes of the path frame. Hence, $x_{b/p}$ is the position of the vehicle along the path-frame tangential axis and $y_{b/f}$ is the position of the vehicle along the path-frame normal axis. That is, the 
path-following problem is solved if we regulate both $x_{b/p}$ and $y_{b/p}$ to zero when $\bsym{p}_{p}(\theta) \triangleq (x_p(\theta), y_p(\theta))$ describes the path $P$ parametrized by $\theta$.
\begin{figure}[!htb]
\centering
\includegraphics[width=.7\columnwidth]{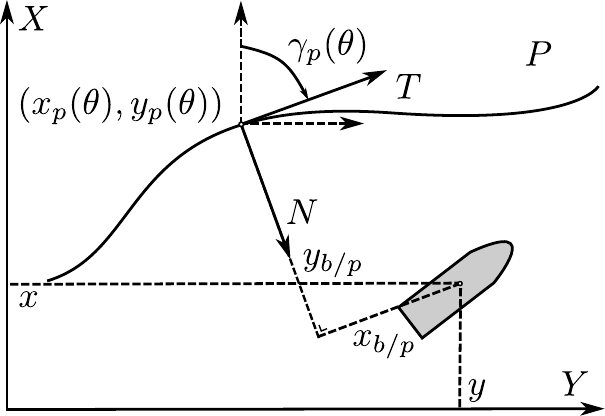}
\caption{Definition of the path.}\label{COG-fig:path}
\end{figure}
\\
The dynamics of the error coordinates introduced in 
\eqref{COG_eq:ye} is computed by substitute (\ref{COG-eq:xdot}-\ref{COG-eq:phidot}) in the derivative of \eqref{COG_eq:ye}. After some rearrangements and basic trigonometric relations we obtain:
\begin{subequations} \label{COG-eq:frame}
\begin{align} 
\dot{x}_{b/p} &= u_t \cos(\chi - \gamma_p)- \dot{\theta} (1 - \kappa(\theta)y_{b/p}) + V_T \label{COG-eq:framea}\\ 
\dot{y}_{b/p} &= u_t \sin(\chi - \gamma_p) + V_N - \kappa(\theta)\dot{\theta}x_{b/p} \label{COG-eq:frameb}
\end{align}
\end{subequations}
where $\chi \triangleq \psi + \beta$ and $\beta \triangleq \arctan(v_r / u_r)$ are the course and the side-slip angles respectively (see Figure \ref{COG-fig:defin}), $V_T \triangleq V_x \cos(\gamma_p(\theta)) + V_y \sin(\gamma_p(\theta))$ and $V_N \triangleq V_y \cos(\gamma_p(\theta)) - V_x \sin(\gamma_p(\theta))$ are the ocean current component in the tangential direction and normal direction of the path-tangential reference frame respectively.
\section{Path Parametrisation}
As proposed in \cite{lapierre2007nonlinear} we can use the update law of the path variable as an extra degree of freedom in the controller design. In particular, the propagation speed of the frame is used to obtain the desired behaviour of the $x_{b/p}$ dynamics. We choose 
\begin{equation} \label{COG-eq:thgbl}
\dot{\theta} = u_{t} \cos(\chi - \gamma_p(\theta)) + k_x f_{\theta}(x_{b/p},y_{b/p}) + V_T
\end{equation}
where $k_x>0$ is a control gain for the convergence of $x_{b/p}$ and $f_{\theta}(x_{b/p},y_{b/p})$ is a function to be designed later satisfying $f_{\theta}(x_{b/p},y_{b/p})x_{b/p} > 0$. Consequently, when substituting \eqref{COG-eq:thgbl} in \eqref{COG-eq:framea} we obtain
\begin{equation}
\dot{x}_{b/p} = - k_x f_{\theta}(x_{b/p},y_{b/p}) + \dot{\theta}\kappa(\theta)y_{b/p}.
\end{equation} 
For the case where the current is unknown we need to replace $V_T$ by its estimate $\hat{V}_T$, and the update law becomes
\begin{equation} \label{COG-eq:thgbl2}
\dot{\theta} = u_{t} \cos(\chi - \gamma_p(\theta)) + k_x f_{\theta}(x_{b/p},y_{b/p}) + \hat{V}_T
\end{equation}
Substituting this revised update law into \eqref{COG-eq:frame} results in
\begin{align}
\dot{x}_{b/p} &= - k_x f_{\theta}(x_{b/p},y_{b/p}) + \dot{\theta}\kappa(\theta)y_{b/p} + \tilde{V}_T \label{COG-eq:dxbp} \\
\dot{y}_{b/p} &= u_{t}\sin(\chi - \gamma_p(\theta)) + V_N - x_{b/p}\kappa(\theta)\dot{\theta} \label{COG-eq:dye}
\end{align} 
where $ \tilde{V}_T \triangleq V_T - \hat{V}_T $. Note that, as opposed to \cite{belleter2017observera}, the parametrization \eqref{COG-eq:thgbl2} does not decouple \eqref{COG-eq:dxbp} from \eqref{COG-eq:dye}. Consequently, since \eqref{COG-eq:dxbp} depends on $y_{b/p}$, the state $x_{b/p}$ does not converge independently from $y_{b/p}$ and both $x_{b/p}$ and $y_{b/p}$ will have to be regulated to zero using the surge and yaw rate controllers. Moreover, note that although this parametrisation has the advantage over \cite{belleter2017observera} that the update law can be well defined on the entire state space, the path-following error is here not defined as the shortest distance to the path since the vehicle is not necessarily on the normal.  
\section{Controllers, Observer, and Guidance} \label{COG-sec:ctrl}
In this section we design the two control laws, $\tau_u$ and $\tau_r$, and the ocean current estimator that are used to achieve path-following. In the first subsection we present the velocity control law $\tau_u$. The second subsection presents the ocean current observer. The third subsection presents the guidance law to be used.
\subsection{Surge velocity control}
The velocity control law is a feedback-linearising P-controller that is used to drive the relative surge velocity to the desired $u_{rd}(t)$ and is given by
\begin{equation} \label{COG-eq:tauu}
\tau_u = -F_{u_r}(v_r,r) + \dot{u}_{rd} + \frac{d_{11}}{m_{11}}u_{rd} - k_u (u_r - u_{rd})
\end{equation}
where $k_u>0$ is a constant controller gain. It is straightforward to verify that \eqref{COG-eq:tauu} ensures global exponential tracking of the desired velocity. In particular, when \eqref{COG-eq:tauu} is substituted in \eqref{COG-eq:urdot} we obtain
\begin{equation} \label{COG-eq:utild}
\dot{\tilde{u}}_r = -k_{u}(u_r-u_{rd}) = -k_u \tilde{u}_r
\end{equation}
where $\tilde{u}_r \triangleq u_r -u_{rd}$. Consequently, the velocity error dynamics are described by a stable linear systems, which assures exponential tracking of the desired velocity $u_{rd}$.

\subsection{Ocean current estimator} \label{COG-subsec:obs}
In this section we present the ocean current observer and show that for constant ocean currents the estimation errors are globally exponentially stable. Moreover, for an appropriate choice of the initial conditions we have
$$ \Vert \hat{V}_N(t) \Vert < u_{rd}(t)\textrm{, if~~} 2 V_{\max} < u_{rd}(t),~\forall t \geq 0$$
where $\hat{V}_N(t)$ is the estimate of $V_N(t)$. \\ 
We will use the ocean current estimator introduced in \cite{moe2014path} and used in \cite{belleter2017observera}. This observer provides the estimate of the ocean current needed to implement \eqref{COG-eq:thgbl2} and the guidance law developed in the next subsection. Rather then estimating the time-varying current components in the path frame $V_T$ and $V_N$, the observer is used to estimate the constant ocean current components in the inertial frame $V_x$ and $V_y$. The observer from \cite{aguiar2007dynamic} is based on the kinematic equations of the vehicle, i.e. \eqref{COG-eq:xdot} and \eqref{COG-eq:ydot}, and requires measurements of the vehicle's $x$ and $y$ position in the inertial frame. The observer is formulated as
\begin{subequations} \label{COG-eq:obs}
\begin{align}
\dot{\hat{x}} &= u_r\cos(\psi) - v_r\sin(\psi) + \hat{V}_x + k_{x_1} \tilde{x} \\
\dot{\hat{y}} &= u_r\sin(\psi) + v_r\cos(\psi) + \hat{V}_y + k_{y_1} \tilde{y} \\
\dot{\hat{V}}_x &= k_{x_2}\tilde{x} \\
\dot{\hat{V}}_y &= k_{y_2}\tilde{y}
\end{align}
\end{subequations}
where $\tilde{x} \triangleq x - \hat{x}$ and $\tilde{y} \triangleq y -\hat{y}$ are the positional errors and $k_{x_1}$, $k_{x_2}$, $k_{y_1}$, and $k_{y_2}$ are constant positive gains. Consequently, the estimation error dynamics are given by
\begin{equation} \label{COG-eq:obserr}
\begin{bmatrix} \dot{\tilde{x}} \\ \dot{\tilde{y}} \\ \dot{\tilde{V}}_x \\ \dot{\tilde{V}}_y \end{bmatrix} = \begin{bmatrix} -k_{x_1} & 0 & 1 & 0 \\ 0 & -k_{y_1} & 0 & 1 \\ -k_{x_2} & 0 & 0 & 0 \\ 0 & -k_{y_2} & 0 & 0 \end{bmatrix} \begin{bmatrix} \tilde{x} \\ \tilde{y} \\ \tilde{V}_x \\ \tilde{V}_y \end{bmatrix}
\end{equation}
which is a linear system with negative eigenvalues. Hence, the observer error dynamics are globally exponentially stable at the origin. Note that this implies that also $\hat{V}_T$ and $\hat{V}_N$ go to $V_T$ and $V_N$ respectively with exponential convergence since it holds that
\begin{subequations} \label{COG-eq:Vtrans}
\begin{align} 
\hat{V}_T &= \hat{V}_x \cos(\gamma(\theta)) + \hat{V}_y \sin(\gamma(\theta)) \\
\hat{V}_N &= -\hat{V}_x \sin(\gamma(\theta)) + \hat{V}_y \cos(\gamma(\theta)).
\end{align}
\end{subequations}
For implementation of the controllers it is desired that $\Vert \hat{V}_N(t) \Vert < u_{rd}(t)~\forall t$. To achieve this we first choose the initial conditions of the estimator as 
\begin{equation} \label{eq:estim}
[\hat{x}(t_0),\hat{y}(t_0),\hat{V}_x(t_0),\hat{V}_y(t_0)]^T = [x(t_0),y(t_0),0,0]^T.
\end{equation}
Consequently, the initial estimation error is given by 
\begin{equation}
[ \tilde{x}(t_0),\tilde{y}(t_0),\tilde{V}_x(t_0),\tilde{V}_y(t_0)]^T = [0,0,V_x,V_y]^T
\end{equation}
which has a norm smaller than or equal to $V_{\max}$ according to Assumption \ref{COG-assum:current}. Now consider the function
\begin{equation}
W(t) = \tilde{x}^2 + \tilde{y}^2 + \frac{1}{k_{x_2}}\tilde{V}^2_x + \frac{1}{k_{y_2}}\tilde{V}^2_y
\end{equation}
which has the following time derivative
\begin{align}
\begin{split}
\dot{W}(t) &= -2k_{x_1}\tilde{x}^2 - 2k_{y_1}\tilde{y}^2 \leq 0.
\end{split}
\end{align}
This implies that $W(t)\leq \Vert W(t_0) \Vert$. From our choice of initial conditions we know that
\begin{equation}
\Vert W(t_0) \Vert = \frac{V^2_x}{k_{x_2}} + \frac{V^2_y}{k_{y_2}} \leq \frac{1}{\min(k_{x_2},k_{y_2})}V^2_{\max}.
\end{equation}
Moreover, it is straightforward to verify
\begin{equation}
\frac{1}{\max(k_{x_2},k_{y_2})}\Vert \tilde{\bsym{V}}_c(t) \Vert^2 \leq W(t).
\end{equation}
Combining the observations given above we obtain
\begin{equation}
\frac{1}{\max(k_{x_2},k_{y_2})}\Vert \tilde{\bsym{V}}_c(t) \Vert^2 \leq \frac{1}{\min(k_{x_2},k_{y_2})}V^2_{\max}.
\end{equation}
Consequently, we obtain
\begin{align}
\begin{split}
\Vert \tilde{\bsym{V}}_c(t) \Vert &\leq \sqrt{\tfrac{\max(k_{x_2},k_{y_2})}{\min(k_{x_2},k_{y_2})}}V_{\max} \\ &< \sqrt{\tfrac{\max(k_{x_2},k_{y_2})}{\min(k_{x_2},k_{y_2})}}u_{rd}(t),~\forall t,
\end{split}
\end{align}
which implies that if the gains are chosen as $k_{x_2}=k_{y_2}$ we have 
\begin{equation}
\Vert \hat{V}_N \Vert \leq 2 V_{\max} \leq u_{rd}(t),~\forall t.
\end{equation}
Hence, $\Vert \hat{V}_N \Vert < u_{rd}(t),~\forall t$ if $2 V_{\max} < u_{rd}(t),~\forall t$. 
\begin{remark}
The bound $u_{rd}(t)>2V_{\max},~\forall t$, is only required when deriving the bound on the solutions of the observer. In particular, it is required to guarantee that $\Vert \hat{V}_N \Vert < u_{rd}(t),~\forall t$. For the rest of the analysis it suffices that $V_{\max} < u_{rd},~\forall t$. Therefore, if the more conservative bound $2V_{\max} < u_{rd},~\forall t$ is not satisfied, the observer can be changed to an observer that allows explicit bounds on the estimate $\hat{V}_N$, e.g. the observer developed in \cite{narendra1987new}, rather than an observer that only provides a bound on the error $\tilde{\bsym{V}}_c$ as is the case here. For practical purposes, the estimate can also be saturated such that $\Vert \hat{V}_N \Vert < u_{rd},~\forall t$, which is the approach taken in \cite{moe2014path}. However, in the theoretical analysis of the yaw controller we use derivatives of $\hat{V}_{N}$ which will be discontinuous when saturation is applied.
\end{remark}
\subsection{Guidance for global parametrisation}
When using the global parametrisation we can define one guidance law that can be used everywhere. As in \cite{moe2014path} we choose a line-of-sight like guidance law of the form:
\begin{equation} \label{COG-eq:psidgl}
\psi_d = \gamma(\theta) - \mathrm{atan}\left(\frac{v_r}{u_{rd}}\right) - \mathrm{atan}\left(\frac{y_{b/p} + g}{\Delta (\bsym{p}_{b/p})} \right).
\end{equation}
The guidance law consists of three terms. The first term is a feedforward of the angle of the path with respect to the inertial frame. The second part is the desired side-slip angle, i.e. the angle between the surge velocity and the total speed when $u_r \equiv u_{rd}$. This side-slip angle is used to make the vehicle's total speed tangential to the path when the sway velocity is non-zero. The third term is a line-of-sight (LOS) term that is intended to steer the vehicle to the path, where $g$ is a term dependent on the ocean current. The choice of $g$ provides an extra design freedom to compensate for the component of the ocean current along the normal axis $V_N$. \\
The term $\Delta (\bsym{p}_{b/p})$ is the look-ahead distance. The look-ahead distance has a constant part and a part that depends on the path-following error $\bsym{p}_{b/p}$, i.e. the distance between the current position of the vehicle and the point on the path defined by the current value of $\theta$.

When we substitute \eqref{COG-eq:psidgl} in \eqref{COG-eq:dye} we obtain
\begin{align} \label{COG-eq:dyegl}
\begin{split}
\dot{y}_{b/p} = &~u_{td} \sin\left(\psi_d +\tilde{\psi} +\beta_d - \gamma_p(\theta)\right) + V_N \\ &~- x_{b/p}\kappa(\theta)\dot{\theta} + \tilde{u}_r \sin(\psi - \gamma_p(\theta)) \\
= &~- \frac{u_{td}(y_{b/p} + g)}{\sqrt{(y_{b/p}+g)^2 + \Delta^2}} - x_{b/p} \dot{\gamma}_p(\theta) \\ &~+ V_N + G_1(\tilde{\psi},\tilde{u}_r,g,\psi_d,y_{b/p},u_{td})
\end{split}
\end{align}
where $G_1(\cdot)$ is a perturbing term given by
\begin{align} \label{COG-eq:Ggl}
\begin{split}
G_1(\cdot ) = &~u_{td}\left[ 1- \cos(\tilde{\psi}) \right]\sin \left( \arctan \left(  \frac{y_{b/p} + g}{\Delta} \right) \right) \\  &~+ \tilde{u}_r \sin(\psi-\gamma_p(\theta)) \\ &+ u_{td}\cos \left( \arctan \left(  \frac{y_{b/p} + g}{\Delta} \right) \right) \sin (\tilde{\psi}).
\end{split}
\end{align} 
Note that $G_1(\cdot)$ satisfies
\begin{subequations} \label{COG-eq:Gbnd}
\begin{align} 
G_1(0,0,g,\psi_d,y_{b/p},u_{td}) &= 0 \\
\Vert G_1(\tilde{\psi},\tilde{u}_r,\psi_d,y_{b/p},u_{td}) \Vert &\leq \zeta(u_{td})\Vert [\tilde{\psi},\tilde{u}_r]^T \Vert
\end{align}
\end{subequations}
where $\zeta(u_{td}) > 0 $, which shows that $G_1(\cdot)$ is zero when the perturbing variables are zero and that it has maximal linear growth in the perturbing variables. 

To compensate for the ocean current component $V_N$, the variable $g$ is now chosen to satisfy the equality
\begin{equation}
u_{td} \frac{g}{\sqrt{\Delta^2 +(y_{b/p} +g)^2}} = \hat{V}_N
\end{equation}
which is a choice inspired by \cite{moe2014path}. In order for $g$ to satisfy the equality above, it should be the solution of the following second order equality
\begin{equation*}
\underbrace{(u^2_{td} - \hat{V}^2_N)}_{-a}\left(\frac{g}{\hat{V}_N}\right)^2 = \underbrace{\Delta^2 + y^2_{b/p}}_c + 2 \underbrace{y_{b/p} \hat{V}_N}_b \left(\frac{g}{\hat{V}_N}\right)
\end{equation*}
hence, we choose $g$ to be
\begin{equation}
g = \hat{V}_{N} \frac{b + \sqrt{b^2 - ac}}{-a}
\end{equation}
which has the same sign as $\hat{V}_N$ and is well defined for $(u^2_{rd} - \hat{V}^2_N) = -a > 0$. Substituting this in \eqref{COG-eq:dyegl} gives
\begin{align}
\begin{split}
\dot{y}_{b/p} = &~-u_{td} \frac{y_{b/p}}{\sqrt{(y_{b/p}+g)^2 + \Delta^2}} - x_{b/p} \dot{\gamma}_p(\theta) \\ &~+ \tilde{V}_N + G_1(\tilde{\psi},\tilde{u},\psi_d,y_{b/p},u_{td}).
\end{split}
\end{align}
By choosing $\dot \theta$ to be:
\begin{align} \label{COG-eq:teta}
\dot{\theta} = u_t \cos(\psi + \beta - \gamma_p(\theta)) + \frac{k_\delta x_{b/p}}{\sqrt{1+x^2_{b/p}}} + \hat{V}_{T}
\end{align} 
and substituting \eqref{COG-eq:teta} in \eqref{COG-eq:dxbp}, we obtain: 
\begin{align} \label{COG-eq:39}
\dot x_{b/p} = & - k_\delta \frac{x_{b/p}}{\sqrt{1 + x^2_{b/p}}} + \dot \theta \kappa(\theta) y_{b/p}  + \tilde{V}_{T}
\end{align}
where $k_{\delta}>0$. We see from \eqref{COG-eq:39} that by the choice of $\dot \theta$, we introduce a stabilising term to the tangential error dynamics by appropriately controlling the propagation of our path-tangential frame. 

The derivative of \eqref{COG-eq:psidgl} is given by
\begin{align} \label{COGeq::dpsidgl}
\dot{\psi}_d = &~\kappa(\theta)\dot{\theta} + \frac{y_{b/p} + g}{\Delta^2 + (y_{b/p} + g)^2} \frac{\partial \Delta}{\partial \bsym{p}_{b/p}} \dot{\bsym{p}}_{b/p} \nonumber \\ &- \frac{\dot{v}_r u_{rd} -  \dot{u}_{rd} v_r}{u^2_{rd} + v^2_r} - \frac{\Delta(\dot{y}_{b/p} + \dot{g})}{\Delta^2 + (y_{b/f} + g)^2} 
\end{align}
%\begin{split} \end{split}
with
\begin{equation}
\dot{g} = \dot{\hat{V}}_N \frac{b+ \sqrt{b^2-ac} }{-a} + \frac{\partial g}{\partial a} \dot{a} +  \frac{\partial g}{\partial b} \dot{b} +  \frac{\partial g}{\partial c} \dot{c}.
\end{equation}
The expression for $\dot{\psi}_d$ contains terms depending on $\dot{y}_{b/p}$ and $\dot{x}_{b/p}$ which depend on $\tilde{V}_N$ and $\tilde{V}_T$, respectively. Consequently, $\dot{\psi}_d$ depends on unknown variables and cannot be used to implement the yaw rate controller. This was not considered in \cite{moe2014path} where the proposed controller contained both $\dot{\psi}_d$ and $\ddot{\psi}_d$. 

Moreover, from \eqref{COGeq::dpsidgl} we see that $\dot{\psi}_d$ contains $\dot{v}_r$, which depends on $r = \dot{\psi}$. Therefore, the yaw rate error $\dot{\tilde{\psi}} \triangleq \dot{\psi} - \dot{\psi}_d$ grows with $\dot{\psi}$, which leads to a necessary condition for a well defined yaw rate error. In particular, the dependence on $r=\dot \psi$ becomes clear when we write out the yaw rate error dynamics:
\begin{align} \label{COG-eq:dpsitilgl}
\begin{split}
\dot{\tilde{\psi}} = &~r \left[ 1 + \frac{ X(u_r) u_{rd} } { u^2_{rd} + v^2_r } - \frac{ 2 v_r X(u_r)\Delta }{\Delta^2 + \left( y_{b/p} + g \right)^2 } \frac{\partial g}{\partial a}\right] \\ 
&-  \kappa(\theta) \dot{\theta} + \frac{ Y(u_r) v_r  u_{rd} - \dot{u}_{rd} v_r } { u^2_{rd} + v^2_r } \\ 
&+ \frac{ 2\Delta }{\Delta^2 + \left( y_{b/p} + g \right)^2 } \left[  \dot{\hat{V}}_{N} \frac{ b + \sqrt{b^2-ac} }{-2a} \right.\\  &+\frac{\partial g}{\partial a} \left( \hat{V}_{N} \dot{\hat{V}}_{N} - u_{rd} \dot{u}_{rd} - v_r Y(u_r) v_r \right) \\ 
&+\frac{\partial g}{\partial b}\dot{\hat{V}}_{N} y_{b/p} + \left[ \tfrac{1}{2} + \frac{\partial g}{\partial c} y_{b/p} + \frac{\partial g}{\partial b} \hat{V}_{N}\right] \dot{y}_{b/p} \\ 
&+ \left. \frac{\partial g}{\partial c} \Delta \left[ \frac{\partial \Delta}{\partial x_{b/p}} \dot{x}_{b/p} + \frac{\partial \Delta}{\partial y_{b/p}} \dot{y}_{p/f} \right] \right] \\
&- \frac{y_{b/p} + g}{\Delta^2 + (y_{b/p} + g)^2} \left[ \frac{\partial \Delta}{\partial x_{b/p}} \dot{x}_{b/p}  + \frac{\partial \Delta}{\partial y_{b/p}} \dot{y}_{b/p} \right] \end{split} \notag \\
\triangleq &~C_r(\cdot) r + f_{\psi}(x_{b/p},y_{b/p},u_r,v_r,\theta).
\end{align}
Since $\dot{\psi}_d$ depends on the unknown signal $\tilde{V}_N$ we cannot choose $r_d = \dot{\psi}_d$. To define an expression for $r_d$ without requiring the knowledge of $\tilde{V}_N$, we define $r_d = f_{\psi}(x_{b/p},y_{b/p},u_r,v_r,\theta)/C_r$. Consequently, we have the following necessary condition for the existence of our controller:
\begin{condition} \label{COG-cond:Cr}
It should hold that
\begin{align} \label{COG-eq:28}
\begin{split}
C_r \triangleq 1 + &\left[\frac{ X(u_r)u_{rd} } { u^2_{rd} + v^2_r } - \frac{ 2X(u_r)v_r\Delta}{\Delta^2 + \left( y_{b/p} + g \right)^2 } \frac{\partial g}{\partial a} \right] 
\end{split}
\end{align}
is larger than zero such that the yaw rate controller is well defined for all time.
\end{condition}
\begin{remark}
The condition above can be verified for any positive velocity, for the vehicles that we have model parameters for. Note that for most vehicles this condition is verifiable since standard vehicle design practices will result in similar properties of the function $X(u_r)$. 
 Besides having a lower bound greater then zero, $C_r$ is also upper-bounded since the term between brackets can be verified to be bounded in its arguments.
\end{remark}
As discussed above, since $\dot{\psi}_d$ depends on the unknown signal $\tilde{V}_N$, we cannot choose $r_d = \dot{\psi}_d$. To define an expression for $r_d$ without requiring the knowledge of $\tilde{V}_N$ we use \eqref{COG-eq:28} to define
\begin{align} \label{COG-eq:rdgl} 
\begin{split}
r_d &=  -C^{-1}_r \left[ \kappa(\theta) \left( u_t \cos(\chi - \gamma_p) +  \tfrac{k_\delta x_{b/p}}{\sqrt{ 1 + x^2_{b/p} }} +\hat{V}_T \right) \right. \\ &+ \tfrac{ Y(u_r) v_r  u_{rd} - \dot{u}_{rd} v_r } { u^2_{rd} + v^2_r } + \tfrac{ \Delta }{\Delta^2 + \left( y_{b/p} + g \right)^2 } \left[  \dot{\hat{V}}_{N} \tfrac{ b + \sqrt{b^2-ac} }{-a} \right.\\ 
&+  2\tfrac{\partial g}{\partial b} \dot{\hat{V}}_{N} y_{b/p} +2\tfrac{\partial g}{\partial a} \left(\hat{V}_{N} \dot{\hat{V}}_{N} - u_{rd} \dot{u}_{rd} \right. \\ &\left. -  Y(u_r) v^2_r \right) + \left[ 1 + \tfrac{\partial g}{\partial c} 2 y_{b/p} + \tfrac{\partial g}{\partial b} 2 \hat{V}_{N} \right] \times \\ &\times \left(    \tfrac{- u_{td} y_{b/p} }{ \sqrt{\Delta^2 + (y_{b/p} + g)^2} } + G_1 - x_{b/p} \kappa(\theta)\dot{\theta} \right) \\ 
&+ \left. \tfrac{\partial g}{\partial c} 2 \Delta \left[ \tfrac{\partial \Delta}{\partial x_{b/p}}  \left( \tfrac{- k_\delta x_{b/p}}{\sqrt{ 1 + x^2_{b/p} }} + y_{b/p}\kappa(\theta)\dot{\theta} \right) \right.\right. \\ 
&+ \left.\left. \tfrac{\partial \Delta}{\partial y_{b/p}} \left( \tfrac{- u_{td} y_{b/p} }{ \sqrt{\Delta^2 + (y_{b/p} + g)^2} } + G_1 - x_{b/p} \kappa(\theta)\dot{\theta} \right) \right] \right] \\
&- \tfrac{y_{b/p} + g}{\Delta^2 + (y_{b/p} + g)^2} \left[ \tfrac{\partial \Delta}{\partial x_{b/p}}  \left( \tfrac{- k_\delta x_{b/p}}{\sqrt{ 1 + x^2_{b/p} }} + y_{b/p}\kappa(\theta)\dot{\theta} \right) \right. \\  
&+ \left.\left. \tfrac{\partial \Delta}{\partial y_{b/p}} \left(  \tfrac{- u_{td} y_{b/p} }{ \sqrt{\Delta^2 + (y_{b/p} + g)^2} } + G_1 - x_{b/p} \kappa(\theta)\dot{\theta}\right) \right]\right]
\end{split}	
\end{align}
with,
\begin{align}
\begin{split}
\dot{\hat{V}}_{N} = &~\dot{\hat{V}}_y \cos(\gamma_p(\theta)) -  \dot{\hat{V}}_x  \sin(\gamma_p(\theta)) + \kappa(\theta)\hat{V}^2_T  \\ &- \hat{V}_{T}\kappa(\theta) \left( u_t \cos(\chi - \gamma_p(\theta)) +  \frac{k_\delta x_{b/p}}{\sqrt{ 1 + x^2_{b/p} }} \right). 
\end{split}
\end{align}
Notice that \eqref{COG-eq:rdgl} is equivalent to \eqref{COG-eq:psidgl}, but without the terms depending on the unknowns $\tilde{V}_x$ and $\tilde{V}_y$ that cannot be used in the control inputs. If we substitute \eqref{COG-eq:rdgl} in \eqref{COG-eq:dpsitilgl} and use $\tilde{r} \triangleq r - r_d$, we obtain 
\begin{align} \label{COG-eq:dpsitilgl2}
\begin{split}
\dot{\tilde{\psi}} = &~C_r \tilde{r} + \frac{ 2\Delta \left[ \tfrac{1}{2} + \tfrac{\partial g}{\partial c} y_{b/p} +  \tfrac{\partial g}{\partial b} \hat{V}_{N}  \right]}{\Delta^2 + \left( y_{b/p} + g \right)^2 }  \tilde{V}_{N} \\  & +\frac{ 2\Delta^2 \partial g/\partial c - (y_{b/p} + g) }{\Delta^2 + \left( y_{b/p} + g \right)^2 } \frac{\partial \Delta}{\partial \bsym{p}_{b/p}}  \begin{bmatrix} \begin{smallmatrix}\tilde{V}_T \\ \tilde{V}_N \end{smallmatrix}\end{bmatrix}.
\end{split} 
\end{align} 
From \eqref{COG-eq:dpsitilgl2} it can be seen that using our choice of $r_d$ results in yaw angle error dynamics that have a term dependent on the yaw rate error $\tilde{r}$ and a perturbing term that vanishes when the estimation errors $\tilde{V}_T$ and $\tilde{V}_N$ go to zero. To add acceleration feedforward to the yaw rate controller, the derivative of $r_d$ should be calculated. From the definition of $r_d$, it can be seen that $r_d$ has the following dependencies $r_d = r_d( h^T, y_{b/p}, x_{b/p}, \tilde{\psi}, \tilde{x}, \tilde{y} )$ with $h \triangleq [\theta, v_r, u_r, u_{rd}, \dot{u}_{rd}, \hat{V}_{T}, \hat{V}_{N} ]^T$ a vector that contains all the variable, whose time derivative do not depend on $\tilde{V}_N$ and $\tilde{V}_T$. However, the other dependencies of $r_d$ do introduce new terms depending on $\tilde{V}_N$ and $\tilde{V}_T$ when the acceleration feedforward is calculated. Consequently, we instead define our yaw rate controller with an acceleration feedforward that contains only the known terms from $\dot{r}_d$
\begin{align} \label{COG-eq:tau_rgl}
\tau_r = &~- F(u_r, v_r, r) + \tfrac{\partial r_d}{\partial h^T} \dot{h} - k_1 \tilde{r} - k_2 C_r \tilde{\psi} \notag \\ \begin{split} &~+\tfrac{\partial r_d}{\partial x_{b/p}} \left( - \tfrac{k_\delta x_{b/p}}{\sqrt{ 1 + x^2_{b/p} }} + y_{b/p} \kappa(\theta) \dot{\theta} \right) \\  &+ \tfrac{\partial r_d}{ \partial \tilde{\psi}} C_r \tilde{r}  - \tfrac{\partial r_d}{ \partial \tilde{x}} k_x \tilde{x} - \tfrac{\partial r_d}{ \partial \tilde{y}} k_y \tilde{y} \end{split} \\ &- \tfrac{\partial r_d}{\partial y_{b/p}} \left( \tfrac{ u_{td}y_{b/p} } { \sqrt{\Delta^2 + (y_{b/p} + g)^2 } } - G_1(\cdot) + x_{b/p} \kappa(\theta) \dot{\theta} \right) \notag
\end{align}
where $k_1>0$ and $k_2>0$ are constant controller gains. 

Using the controller \eqref{COG-eq:tau_rgl} in \eqref{COG-eq:rdot} the yaw rate error dynamics become
\begin{align} 
\dot{\tilde{r}} = & - k_1 \tilde{r} - k_2 C_r\tilde{\psi} + \tfrac{\partial r_d}{\partial \tilde{x}} \tilde{V}_x + \tfrac{\partial r_d}{ \partial \tilde{y}} \tilde{V}_y \label{COG-eq:drtilgl2} \\ 
&- \tfrac{\partial r_d}{ \partial \tilde{\psi}} \left[ \tfrac{ 2\Delta}{\Delta^2 + \left( y_{b/p} + g \right)^2 } \left[ \tfrac{1}{2} + \tfrac{\partial g}{\partial c} y_{b/p} +  \tfrac{\partial g}{\partial b} \hat{V}_{N}  \right] \tilde{V}_{N} \right. \notag \\  
& \left.+ \tfrac{ 2\Delta^2 \partial g/\partial c - (y_{b/p} + g) }{\Delta^2 + \left( y_{b/p} + g \right)^2 } \tfrac{\partial \Delta}{\partial \bsym{p}_{b/p}} \begin{bmatrix} \begin{smallmatrix}\tilde{V}_T \\ \tilde{V}_N \end{smallmatrix}\end{bmatrix} \right] - \tfrac{\partial r_d}{\partial \bsym{p}_{b/p}} \begin{bmatrix} \begin{smallmatrix}\tilde{V}_T \\ \tilde{V}_N \end{smallmatrix}\end{bmatrix} \notag
\end{align}
which contains two stabilising terms $- k_1 \tilde{r}$ and $ - k_2 C_r \tilde{\psi}$, and perturbing terms depending on $\tilde{V}_T$ and $\tilde{V}_N$ that cannot be cancelled by the controller. 
\begin{remark}
It is straightforward to verify that all the terms in \eqref{COG-eq:rdgl} are smooth fractionals that are bounded with respect to $(y_{b/p}$, $x_{b/p}$, $\tilde{x}$, $\tilde{y}$, $\tilde{\psi}$, $\Delta)$ or are periodic functions with linear arguments, and consequently the partial derivatives in \eqref{COG-eq:tau_rgl} and \eqref{COG-eq:drtilgl2} are all bounded.
This is something that is used when showing closed-loop stability in the next section.
\end{remark}
\section{Closed-Loop Analysis} \label{COG-sec:clsys}
In this section we analyse the closed-loop system of the model \eqref{COG-eq:dynsys} with controllers \eqref{COG-eq:tauu} and \eqref{COG-eq:tau_rgl} and observer \eqref{COG-eq:obs}, when the frame propagates along the path $P$ with update law \eqref{COG-eq:thgbl2}. 
To show that the path following is achieved we have to show that $x_{b/p}$ and $y_{b/p}$ converge to zero, and the closed-loop error dynamics of $\tilde{u}$, $\tilde{\psi}$, and $\tilde{r}$ also converge to zero. However, for for the sway velocity, since we consider general curved paths, the best hope is to be able to show global boundedness. 
Since the observer and the surge velocity dynamics converge independently of the other variables, we define two sets of variables: $\tilde{X}_1 \triangleq [y_{b/p},x_{b/p},\tilde{\psi},\tilde{r}]^T$ and $\tilde{X}_2 \triangleq [\tilde{x},\tilde{y},\tilde{V}_x,\tilde{V}_y,\tilde{u}]^T$ where $\tilde{X}_2$ contains all the variables that converge to zero independently of the others. Moreover, while the variables $\tilde X_1$ and $\tilde X_2$ should converge to zero, the sway velocity is required to remain bounded. 

To show that the error variables in $\tilde{X}_1$ and $\tilde{X}_2$ converge to zero, we consider the closed-loop system:
\begin{subequations} \label{COG-eq:35}
\begin{align} 
&\dot {\tilde{X}}_1 = \begin{bmatrix} \frac{- u_{td}  y_{b/p} }{ \sqrt{\Delta^2 + (y_{b/p} + g)^2} } - x_{b/p} \kappa(\theta)\dot{\theta} + G_1(\cdot) \\ \frac{- k_{\delta} x_{b/p}}{\sqrt{ 1 + x^2_{b/p} }} + y_{b/p}\kappa(\theta)\dot{\theta} \\ C_r \tilde{r} \\ - k_1 \tilde{r} - k_2C_r \tilde{\psi} \end{bmatrix} \label{COG-eq:35a}\\ 
&~+ \begin{bmatrix} \tilde{V}_{N} \\ \tilde{V}_{T} \\ G_2(\Delta,y_{b/p},x_{b/p},g,\hat{V}_N,\hat{V}_T,\tilde{V}_N,\tilde{V}_T)  \\ 
 - \frac{\partial r_d}{ \partial \tilde{\psi}} G_2(\cdot) - \tfrac{\partial r_d}{\partial \bsym{p}_{b/p}} \begin{bmatrix} \begin{smallmatrix}\tilde{V}_T \\ \tilde{V}_N \end{smallmatrix}\end{bmatrix} + \frac{\partial r_d}{\partial \tilde{x}} \tilde{V}_x + \frac{\partial r_d}{ \partial \tilde{y}} \tilde{V}_y  \end{bmatrix} \notag   \\ 
&\dot {\tilde{X}}_2 =  \begin{bmatrix} -k_{x_1} \tilde{x} - \tilde{V}_{x} \\  -k_{y_1} \tilde{y} - \tilde{V}_{y} \\  -k_{x_2} \tilde{x} \\  -k_{y_2} \tilde{y} \\ -k_u \tilde{u} \end{bmatrix}  \label{COG-eq:35b} \\ 
\begin{split} \label{COG-eq:35c}
&\dot{v}_r = X ( u_{rd} + \tilde{u} ) r_d ( h, y_{b/p}, x_{b/p}, \tilde{\psi}, \tilde{x}, \tilde{y}) \\ &~\hspace{5mm}+ X( u_{rd} + \tilde{u} ) \tilde{r} + Y( u_{rd} + \tilde{u} ) v_r   \end{split} 
\end{align} 
\end{subequations}
where
\begin{align}
\begin{split}
G_2(\cdot) = &~\tfrac{ 2\Delta}{\Delta^2 + \left( y_{b/p} + g \right)^2 } \left[ \tfrac{1}{2} + \tfrac{\partial g}{\partial c} y_{b/p} +  \tfrac{\partial g}{\partial b} \hat{V}_{N}  \right] \tilde{V}_{N} \\  
& + \tfrac{ 2\Delta^2 \partial g/\partial c - (y_{b/p} + g) }{\Delta^2 + \left( y_{b/p} + g \right)^2 } \tfrac{\partial \Delta}{\partial \bsym{p}_{b/p}} \begin{bmatrix} \begin{smallmatrix}\tilde{V}_T \\ \tilde{V}_N \end{smallmatrix}\end{bmatrix}.
\end{split}
\end{align}
Note that $G_2(\Delta,y_{b/p},x_{b/p},g,\hat{V}_N,\hat{V}_T,\tilde{V}_N,\tilde{V}_T)$ satisfies
\begin{gather*}
G_2(\Delta,y_{b/p},x_{b/p},g,\hat{V}_N,\hat{V}_T,0,0) = 0 \\
\Vert G_2(\cdot) \Vert \leq \zeta_2(\Delta)\Vert [\tilde{V}_T,\tilde{V}_N] \Vert,
\end{gather*}
where $\zeta_2(\Delta) > 0 $. \\ 
We design the time-varying look-ahead distance as   
\begin{align} \label{COG-eq:Delta}
 \Delta (x_{b/p},y_{b/p}) = & \sqrt{\mu + x^2_{b/p} + y^2_{b/p}} 
\end{align}
where $\mu > 0$ is a constant. Choosing $\Delta$ to depend on $x_{b/p}$ and $y_{b/p}$ is necessary to find a bounded value of $\mu$ to assure local boundedness of $v_r$ with respect to $\tilde{X}_2$ independently of $\tilde{X}_1$.
This shows that $G_2(\cdot)$ is zero when the perturbing variables, i.e. $\tilde{V}_T$ and $\tilde{V}_N$, are zero and $\zeta_2(\Delta)$ has at most linear growth with respect to $x_{b/p}$ and $y_{b/p}$. 

The following three steps are taken by formulating and proving three lemmas. For the sake of brevity in the main body of this paper, the proofs of the following lemmas are replaced by a sketch of each proof in the main body. The full proofs can be found in \citep{proofs-global}.

The first step in the stability analysis of \eqref{COG-eq:35} is to assure that the closed-loop system is forward complete and that the sway velocity $v_r$ remains bounded. Therefore, under the assumption that Condition \ref{COG-cond:Cr} is satisfied, i.e. $ C_r > 0 $, we take the following three steps:
\begin{enumerate}
\item First, we prove that the trajectories of the closed-loop system are forward complete. 
\item Then, we derive a necessary condition such that $v_r$ is locally bounded with respect to $( \tilde{X}_1, \tilde{X}_2 )$.
\item Finally, we establish that for a sufficiently big value of $\mu$, $v_r$ is locally bounded only with respect to $\tilde{X}_2$, i.e. independently of $\tilde{X}_1$. 
\end{enumerate}
\begin{lemma} [Forward completeness] \label{COG-lem1}
The trajectories of the closed-loop system \eqref{COG-eq:35} are forward complete.
\end{lemma}
The proof of this lemma is given in the Appendix. The general idea is as follows. Forward completeness for \eqref{COG-eq:35b} is evident since this part of the closed-loop system consists of GES error dynamics. Using the forward completeness and in fact boundedness of \eqref{COG-eq:35b}, we can show forward completeness of \eqref{COG-eq:35c}, $\dot{\tilde{\psi}}$, and $\dot{\tilde{r}}$. Hence, forward completeness of \eqref{COG-eq:35} depends on forward completeness of $\dot{x}_{b/p}$ and $\dot{y}_{b/p}$. To show forward completeness of $\dot{x}_{b/p}$ and $\dot{y}_{b/p}$, we consider the $x_{b/p}$ and $y_{b/p}$ dynamics with $\tilde{X}_2$, $\tilde{\psi}$, $\tilde{r}$, and $v_r$ as inputs which allows us to show forward completeness of $\dot{x}_{b/p}$ and $\dot{y}_{b/p}$ according to \cite[Corollary 2.11]{angeli1999forward}. Consequently, all the states of the closed-loop system are forward complete, and hence the closed-loop system \eqref{COG-eq:35} is forward complete

\begin{lemma} [Boundedness near $(\tilde{X}_1, \tilde{X}_2)=0$] \label{COG-lem2}
The system \eqref{COG-eq:35c} is bounded near the manifold $(\tilde{X}_1, \tilde{X}_2)=0$ if and only if the curvature of $P$ satisfies the following condition:
\begin{equation} \label{COG-eq:curvlem2}
\kappa_{\max} \triangleq \max_{\theta \in P}\left| \kappa(\theta) \right| < \frac{ Y_{\min}}{2X_{\max} } ~~~~~ X_{max} \triangleq | X(u_r) |_\infty .
\end{equation}
\end{lemma}

The proof of this lemma is given in the Appendix. A sketch of the proof is as follows. The sway velocity dynamics \eqref{COG-eq:35c} are analysed using a quadratic Lyapunov function $V=1/2v^2_r$. It can be shown that the derivative of this Lyapunov function satisfies the conditions for boundedness when the solutions are on or close to the manifold where $(\tilde{X}_1, \tilde{X}_2)=0$. Consequently, \eqref{COG-eq:35c} satisfies the conditions of boundedness near $(\tilde{X}_1,\tilde{X}_2)=0$ as long as \eqref{COG-eq:curvlem2} is satisfied.
\begin{remark}
In the proof of Lemma \ref{COG-lem2} it is shown that choosing $\Delta (x_{b/p}) =\sqrt{\mu + x^2_{b/p}}$ as in \cite{moe2014path}, $v_{r}$ would grow linearly and unbounded with respect to the state $y_{b/f}$.
The necessity of choosing $\Delta$ as in \eqref{COG-eq:Delta}, i.e. $\Delta$ dependent also on $y_{b/f}$, is shown and justified.
\end{remark}
In Lemma \ref{COG-lem2} we show boundedness of $v_r$ for small values of $(\tilde{X}_1, \tilde{X}_2)$ to derive the bound on the curvature. However, locality with respect to $\tilde{X}_1$, i.e. the path-following errors and yaw angle and yaw rate errors, is not desirable, and in the next lemma boundedness independent of $\tilde{X}_1$ is shown under an extra condition on the constant $\mu$ from the definition \eqref{COG-eq:Delta} of the look-ahead distance $\Delta$.
\begin{lemma} [Boundedness near $\tilde{X}_2=0$] \label{COG-lem3}
The system \eqref{COG-eq:35c} is bounded near the manifold $ \tilde{X}_2 = 0$, independently of $\tilde{X}_1$, if we choose
\begin{gather} 
\mu > \frac{ 8 X_{\max} }{ Y_{\min} -2X_{\max} \kappa_{\max}} \label{COG-eq:mu} 
\end{gather}
where $ X_{max} = | X(u_r) |_\infty $ and $\kappa_{\max} = \max_{\theta \in P}\left| \kappa(\theta) \right| $.
\end{lemma} 
The proof of this lemma is given in the Appendix. It follows along the same lines of the proof of Lemma \ref{COG-lem2}. That is, when the solutions are close to the manifold $\tilde{X}_2 = 0$,  rather than $(\tilde{X}_1, \tilde{X}_2)=0$, the boundedness can still be shown provided that \eqref{COG-eq:mu} is satisfied additionally to the conditions of Lemma \ref{COG-lem2}. 
\begin{theorem}
Consider a $\theta$-parametrised path denoted by $ P(\theta) \triangleq (x_p(\theta), y_p(\theta)) $, with the update law given by \eqref{COG-eq:teta}. Then under Condition \ref{COG-cond:Cr} and the conditions of Lemma \ref{COG-lem1}-\ref{COG-lem3}, the system \eqref{COG-eq:dynsys} with control laws \eqref{COG-eq:tauu} and \eqref{COG-eq:tau_rgl} and observer \eqref{COG-eq:obs} follows the path $P$, while maintaining $v_r$, $\tau_r$, and $\tau_u$ bounded. In particular, the origin of the system \eqref{COG-eq:35a}-\eqref{COG-eq:35b} is GAS and LES. 
\end{theorem}
\begin{proof}
From the fact that the origin of \eqref{COG-eq:35b} is GES, the fact that the closed-loop system \eqref{COG-eq:35} is forward complete according to Lemma \ref{COG-lem1}, and the fact that solutions of \eqref{COG-eq:35c} are locally bounded near $\tilde{X}_2 = 0$ according to Lemma \ref{COG-lem3}, we can conclude that there is a finite time 
$ T > t_0 $ after which solutions of \eqref{COG-eq:35b} will be sufficiently close to $\tilde{X}_2 = 0$ to guarantee boundedness of $v_r$. 

Having established that $v_r$ is bounded we first analyse the cascade 
\begin{subequations} \label{COG-eq:clstab}
\begin{align} 
\begin{split}
&\begin{bmatrix} \dot{\tilde{\psi}} \\  \dot{\tilde{r}} \end{bmatrix}  =  \begin{bmatrix}  C_r \tilde{r} \\ - k_1 \tilde{r} - k_2 C_r \tilde{\psi} \end{bmatrix} \\ &\qquad+ 
\begin{bmatrix} G_2(\cdot)  \\ \frac{\partial r_d}{\partial [\tilde{x},\tilde{y}]^T} \tilde{\bsym{V}}_c - \frac{\partial r_d}{ \partial \tilde{\psi}} G_2(\cdot) - \frac{\partial r_d}{\partial \bsym{p}_{b/p}} \begin{bmatrix} \begin{smallmatrix}\tilde{V}_T \\ \tilde{V}_N \end{smallmatrix}\end{bmatrix} \end{bmatrix} \end{split} \label{COG-eq:clstaba}\\
&\begin{bmatrix}  \dot{\tilde{x}} \\ \dot{\tilde{y}} \\ \dot{ \tilde{V}}_x \\ \dot{ \tilde{V}}_y \\ \dot{\tilde{u}} \end{bmatrix} =  \begin{bmatrix} -k_{x_1} \tilde{x} - \tilde{V}_{x} \\  -k_{y_1} \tilde{y} - \tilde{V}_{y} \\  -k_{x_2} \tilde{x} \\  -k_{y_2} \tilde{y} \\ -k_u \tilde{u} \end{bmatrix}.  \label{COG-eq:clstabb}
\end{align} 
\end{subequations}
The perturbing system \eqref{COG-eq:clstabb} is GES as shown in Subsection \ref{COG-subsec:obs}. The interconnection term, i.e. the second matrix in \eqref{COG-eq:clstaba}, satisfies the linear growth criteria from \cite[Theorem~2]{panteley1998global}. More specifically, it does not grow with the states $\tilde{\psi}$ and $\tilde{r}$ since all the partial derivatives of $r_d$ and $G_2(\cdot)$ can respectively be bounded by constants and linear functions of $\tilde{V}_x$ and $\tilde{V}_y$. The nominal dynamics, i.e. the first matrix in \eqref{COG-eq:clstaba}, can be analysed with the following quadratic Lyapunov function
\begin{equation} \label{COG-eq:V1thm}
V_{(\tilde{r},\tilde{\psi})} = \frac{1}{2}\tilde{r}^2 + \frac{1}{2}k_2\tilde{\psi}^2
\end{equation}
whose derivative along the solutions of the nominal system is given by
\begin{equation} \label{COG-eq:dV1thm}
\dot{V}_{(\tilde{r},\tilde{\psi})} = k_2C_r\tilde{r}\tilde{\psi} - k_1\tilde{r}^2 -k_2C_r\tilde{\psi}\tilde{r} = -k_2 \tilde{r}^2 \leq 0
\end{equation}
which implies that $\tilde{r}$ and $\tilde{\psi}$ are bounded. The derivative of \eqref{COG-eq:dV1thm} is given by
\begin{equation}
\ddot{V}_{(\tilde{r},\tilde{\psi})} = -2k^2_1 \tilde{r}^2 - 2k_1k_2C_r \tilde{\psi}\tilde{r}
\end{equation}
which is bounded since $\tilde{r}$ and $\tilde{\psi}$ are bounded. This implies that \eqref{COG-eq:dV1thm} is a uniformly continuous function. 
Note that the nominal system is non autonomous, since $C_r$ depends on the time-varying signals $u_r, u_{rd}, \Delta, g, a$, which are well-defined due to the forward completeness property. We will thus apply Barbalat's lemma to further investigate the stability of the nominal system. We conclude
\begin{equation}
\lim_{t\rightarrow\infty} \dot{V}_{(\tilde{r},\tilde{\psi})} = \lim_{t\rightarrow\infty} -k_1\tilde{r}^2 = 0~\Rightarrow~\lim_{t \rightarrow \infty} \tilde{r} = 0.
\end{equation}
Since $C_r$ is persistently exciting, which follows from the fact that $C_r$ is upper bounded and lower bounded by a constant larger then zero, it follows from the expression of the nominal dynamics that
\begin{equation}
\lim_{t \rightarrow \infty} \tilde{r} = 0~\Rightarrow~\lim_{t \rightarrow \infty} \tilde{\psi} = 0.
\end{equation}
This implies that the system is globally asymptotically stable, and since the nominal dynamics are linear it follows that the nominal dynamics are globally exponentially stable. Consequently, from the above it follows that the cascade \eqref{COG-eq:clstab} is GES using \cite[Theorem~2]{panteley1998global} and \cite[Definition~2.2]{angeli2000characterization}.

We now consider the following dynamics
\begin{align} \label{COG-eq:clstab2}
\begin{split}
\begin{bmatrix} \dot{y}_{b/p} \\ \dot{x}_{b/p} \end{bmatrix}  =  &~\begin{bmatrix}  - u_{td}   \frac{ y_{b/p} }{ \sqrt{\Delta^2 + (y_{b/p} + g)^2} } - x_{b/p} \kappa(\theta)\dot{\theta} \\  - k_\delta \frac{x_{b/p}}{\sqrt{ 1 + x^2_{b/p} }} + y_{b/p}\kappa(\theta)\dot{\theta} \end{bmatrix} \\ &~+ \begin{bmatrix} \tilde{V}_{N} + G_1(\cdot) \\ \tilde{V}_{T} \end{bmatrix}.
\end{split}
\end{align}
Note that we can view the systems \eqref{COG-eq:clstab} and \eqref{COG-eq:clstab2} as a cascaded system where the nominal dynamics are formed by the first matrix of \eqref{COG-eq:clstab2}, the interconnection term is given by second matrix of \eqref{COG-eq:clstab2}, and the perturbing dynamics are given by \eqref{COG-eq:clstab}. As we have just shown, the perturbing dynamics are GES. Using \eqref{COG-eq:Gbnd} it is straightforward to verify that the interconnection term satisfies the conditions of \cite[Theorem~2]{panteley1998global}. We now consider the following Lyapunov function for the nominal system
\begin{equation}
V_{\bsym{p}_{b/p}} = \tfrac{1}{2}x^2_{b/p} + \tfrac{1}{2}y^2_{b/p}
\end{equation}
whose derivative along the solutions of the nominal system is given by
\begin{equation}
\dot{V}_{\bsym{p}_{b/p}} = \frac{- u_{td} y^2_{b/p} }{ \sqrt{\Delta^2 + (y_{b/p} + g)^2} } -  \frac{k_\delta x^2_{b/p}}{\sqrt{ 1 + x^2_{b/p} }}
\end{equation}
is negative definite. The nominal system is thus GAS. Moreover, since it is straightforward to verify that $\dot{V}_{\bsym{p}_{b/p}} \leq \alpha V_{\bsym{p}_{b/p}}$ for some constant $\alpha$ dependent on initial conditions, it follows from the comparison lemma (\cite[Lemma 3.4]{khalil2002nonlinear}) that the nominal dynamics are also LES. Consequently, the cascaded system satisfies the conditions of \cite[Theorem~2]{panteley1998global} and \cite[Lemma~8]{panteley1998exponential}, and therefore the cascaded system is GAS and LES. This implies that the origin of the error dynamics, i.e. $(\tilde{X}_1,\tilde{X}_2)=(0,0)$, is globally asymptotically stable and locally exponentially stable. 
\end{proof} 

\begin{remark}
Note that this proof uses the theory for cascaded systems which is an approach that has also been taken in most previous works concerning this topic. However, the cascaded argument alone would not hold without establishing forward completeness of the closed-loop solutions and some robustness properties with respect to some vanishing variables in the system, see Lemmas \ref{COG-lem1}, \ref{COG-lem2}, and \ref{COG-lem3}. This is a caveat in the stability proof of previous works which we have intended to fill here. Moreover, using the cascade \eqref{COG-eq:clstab} this proof shows that the system can be controlled by a yaw rate controller that does not depend on $\dot{\psi}_d$ and consequently does not depend on the ocean current.
\end{remark} 

\section{Case Study} \label{COG-sec:case}
This section presents two case studies that illustrate the theoretical results presented in this paper. 
In the first case study, the desired path is a sine curve while in the second case study the ship is requested to follow a path composed of two non-co-linear straight lines (zig-zag path).  Moreover, in the latter case, we assume that the ocean current changes direction and magnitude when the ship moves to the second straight-line segment.
In both the case studies, we use the parameters values of the underactuated surface vessel studied in \cite{caharija2014thesis}.The vessel considered in \cite{caharija2014thesis} is an offshore supply vessel equipped with a propeller for thrust and a rudder for yaw actuation. The control input is therefore given by the surge thrust $T_u$ and rudder angle $T_r$ which are allocated as in \eqref{COG_eq:relVelMod}-\eqref{COG_eq:B}. 

\subsection{Sinusoidal path}
The ocean current components are given by $V_x = -0.4~\munit{m/s}$ and $V_y = 1~\munit{m/s}$ and consequently $V_{\max} \approx 1.08~\munit{m/s}$. The desired relative surge velocity is chosen to be constant and set to $u_{rd}= 5~\munit{m/s}$, which means that Assumption \ref{COG-assum:vel} is satisfied. We want to remark that a surge speed of $5~\munit{m/s}$ may be considered as not purely low speed, which could oppose the assumption of linear damping that we made in \eqref{COG_eq:relVelMod}. However, in our simulations we use the damping parameters given in \cite[Appendix B]{caharija2014thesis} which are the result of a linear approximation of the damping term and the approximated linear damping is valid for $|u_{r}|<7~\munit{m/s}$.
Using the model's parameters given in \cite{caharija2014thesis}, we have $(Y_{\min})/(2X_{\max}) \approx 0.0667$. The observer is initialized as in \eqref{eq:estim} and the observer gains are selected as $k_{x_1} = k_{y_1} = 1$ and $k_{x_2} = k_{y_2} = 0.1$. The controller gains are selected as $k_{u_r} = 0.1$ for the surge velocity controller and $k_{1} = 40$ and $k_2 = 100$ for the yaw rate controller. In this first case study the vessel is required to follow the sinusoidal path $ P = \left\{ \bsym{p}_{b/p} \in \mathbb{R}^2 : y_{b/p} = 300\sin\left(\frac{\pi}{800}x_{b/p}\right) \right\} $. Consequently, the maximum curvature of the path is  $\max_{\bsym{p}_{b/p} \in P} |\kappa(\theta(x_p))| = 0.0087$. 
This implies that we satisfy  our constraint on the curvature given by Lemma \ref{COG-lem2} since $\max_{\bsym{p}_{b/p} \in P} |\kappa(\theta(x_p))| < (Y_{\min})/(2X_{\max})$. The required value for $\mu$ can be calculated as suggested in Lemma \ref{COG-lem3} to obtain $\mu > 987.3~\munit{m}$, which can be satisfied by choosing $\mu = 1000~\munit{m}$. 
The initial conditions are
\begin{equation}
[u_r,v_r,r,x,y,\psi]^T = [0,0,0,10,200,\pi/2]^T.
\end{equation}
The resulting motion of the ship are shown in Figure \ref{COG-fig:Glo_path}. The dashed blue line is the trajectory of the vessel and the red sine curve is the reference. The yellow ship shows the orientation of the ship each $100~\munit{s}$. From Figure \ref{COG-fig:Glo_path} it can clearly be seen that the orientation of the ship is not tangent to the sine curve, which is as expected and desired for underactuated vessels in order to compensate for the ocean current. 
\begin{figure}[tbh]
\centering
\includegraphics[width=0.5\columnwidth]{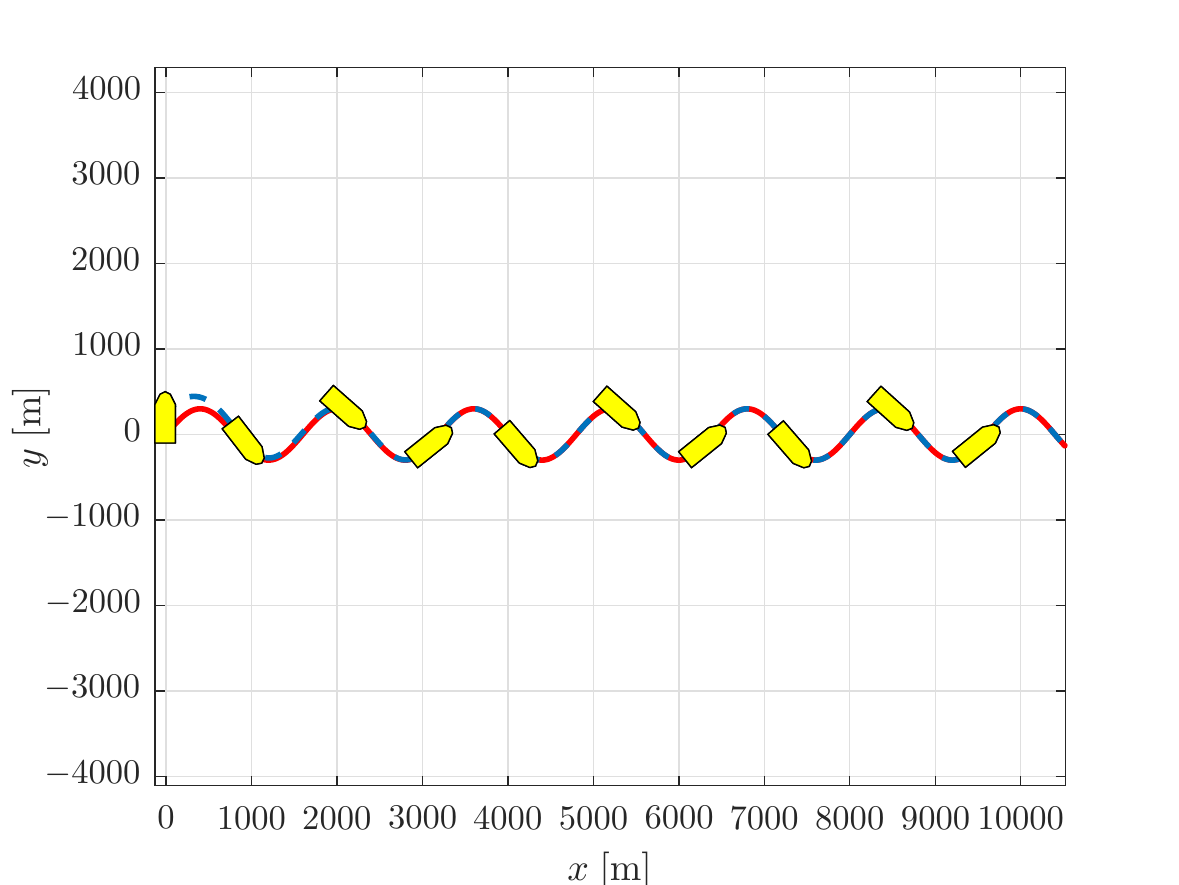}{}
\caption[Path following of the desired path.]{Path following of the desired sinusoidal path in the $x-y$-plane using the proposed controller.}\label{COG-fig:Glo_path}
\end{figure}
The path-following errors in the tangential direction, $x_{b/p}$, and in the normal direction, $y_{b/p}$ can be seen in the top plot of Figure \ref{COG-fig:Glo_sub_cp}, from which it can clearly be seen that the path-following errors converge to zero after a transient period. A detail of the last portion of the simulation is given to illustrate that the errors converge to zero. The estimates of the ocean current components obtained from the ocean current observer are given in the second plot from the top in Figure \ref{COG-fig:Glo_sub_cp}. 
The previous plot illustrates the conservativeness of the bound $2V_{\max} < u_{rd}(t),~\forall t$ derived in the analysis of the observer-error dynamics in Subsection \ref{COG-subsec:obs}. The sway velocity $v_r$ is plotted in the third plot of Figure \ref{COG-fig:Glo_sub_cp}. This plot shows that due to the curvature of the path, the sway velocity does not converge to zero but remains bounded and follows a periodic motion induced by the periodicity of the desired path which has a curvature that both non-constant and non-zero. The relative surge velocity is plotted in the fourth plot of Figure \ref{COG-fig:Glo_sub_cp}. This plot clearly shows the exponential convergence of the velocity as it moves to the desired value of $u_{rd} = 5~\munit{m/s}$. Especially interesting is the coupling of the relative surge velocity with the value of $C_r$ in Condition \ref{COG-cond:Cr}. The parameter $C_r$ is plotted in the bottom plot of Figure \ref{COG-fig:Glo_sub_cp}. From this plot, it can clearly be seen that $C_r$ is bounded away from zero throughout the motion. 
\begin{figure}[tbh]
\centering
\includegraphics[width=0.5\columnwidth]{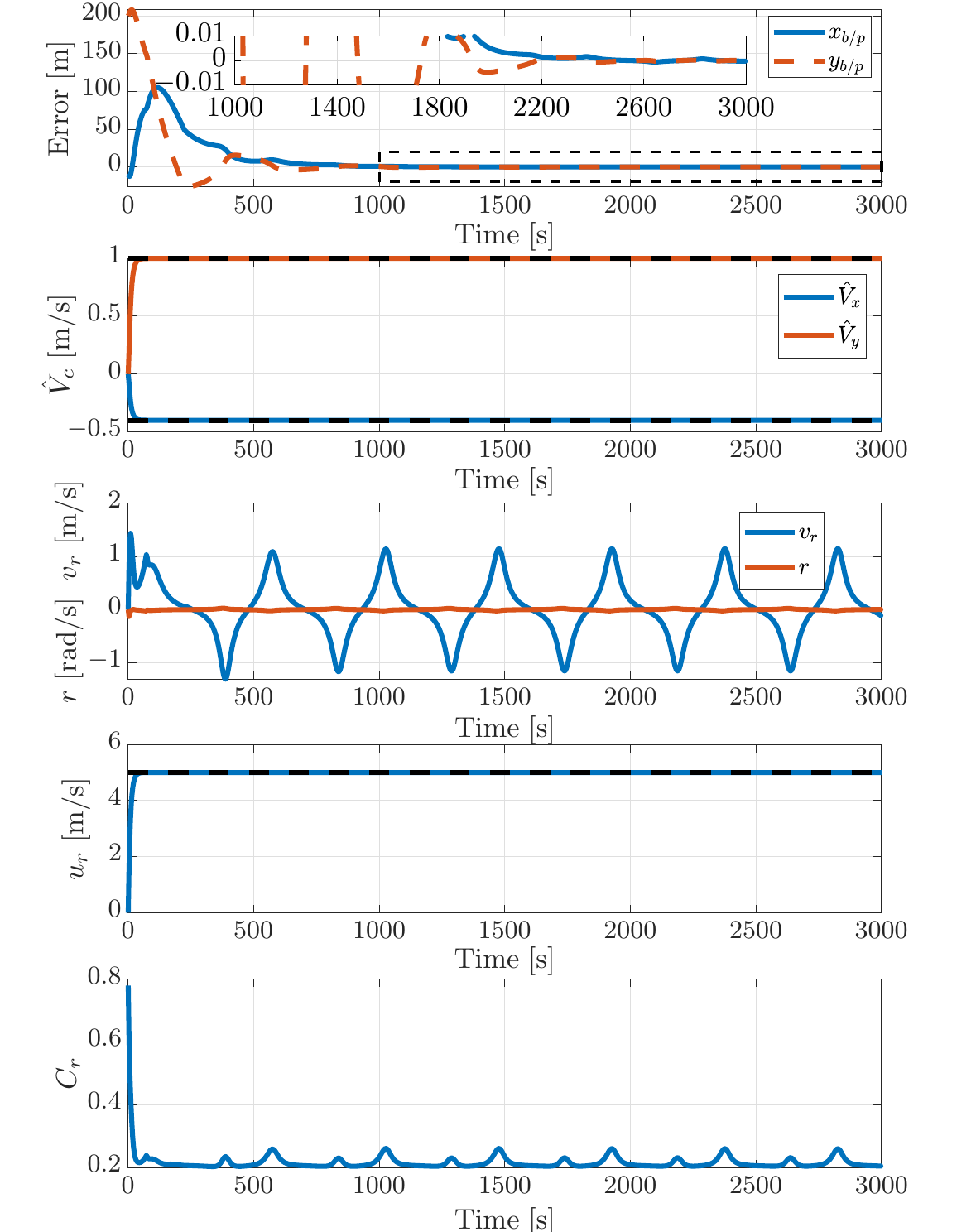}{}
\caption[The time evolution of the path-following errros, current estimates, sway velocity, surge velocity, and the parameter $C_r$.]{Time evolution of the path-following errors (top), current estimates (second), sway velocity, yaw rate (third), surge velocity (fourth), and the parameter $C_r$ (bottom).}\label{COG-fig:Glo_sub_cp}
\end{figure}  
\subsection{Zig-zag path}
In this case study, we assume that the desired path is composed of two linked non-co-linear straight lines, namely $P = P_1 \cup P_2$ where 
$$ P_1 \triangleq \left\{ \bsym{p}_{b/p} \in \mathbb{R}^2 : y_{b/p} =  f(x_{b/p}) = x_{b/p},~~x_{b/p} \in [0,1500~\munit{m}] \right\} $$ 
and 
$$P_2 \triangleq \left\{ \bsym{p}_{b/p} \in \mathbb{R}^2 : y_{b/p} =  f(x_{b/p}) = -x_{b/p}+3000~\munit{m},~~x_{b/p} \in [1500~\munit{m},2600~\munit{m}] \right\}. $$
When the vessel switches between the paths the ocean current switches as well. That is, when $\bsym{p}_{b/p} \in P_1$, the ocean current components are given by $V_{x1}~=~-0.4~\munit{m/s}$ and $V_{y1} = 1.0~\munit{m/s}$ and when the ship maneuvers to follow $P_2$, $\bsym{p}_{b/p} \in P_2$, the ocean current switches its value to $V_{x2} = -1.0~\munit{m/s}$ and $V_{y2} = 0.7~\munit{m/s}$. The simulation with switching current is included to illustrate that despite the assumption on the ocean current in Assumption \ref{COG-assum:current}, the theoretical analysis and the stability conclusions remain valid when the ocean current is constant at sufficiently large intervals of time and varies only over a finite number of sufficiently small intervals of time. This will allow the observer to adapt and estimate the new value of the ocean during the interval of time where it is constant, and hence the guidance law will be able to control the ship towards the path. With the aforementioned values of the ocean current we have $V_{\max} \approx 1.22~\munit{m/s}$. The desired relative surge velocity is constant with $u_{rd}= 5~\munit{m/s}$, which means that Assumption \ref{COG-assum:vel} is satisfied. We remark that also for this simulation we use the damping parameters given in \cite{caharija2014thesis} which give a linear damping term valid for speeds $| u_{r} |<7~\munit{m/s}$.
The observer's initial condition and gains are the same as in the sinusoidal-path case. The control gains are  $k_{u_r} = 0.1$, $k_{1} = 40$ and $k_2 = 100$. In this second case study, the vessel is required to follow the path $P \triangleq P_1 \cup P_2$. Consequently, the curvature of the path is equal to zero for almost all $(x_{b/p}, y_{b/p}) \in P$. Hence, our constraint on the curvature $|\kappa(\theta(x_p))| < (Y_{\min})/(2X_{\max}) \approx 0.0087$ is trivially satisfied almost everywhere since $(Y_{\min})/(2X_{\max})>0$. Furthermore, according to Lemma \ref{COG-lem3}, we need to satisfy $\mu > 460.9~\munit{m}$, which is the case when $\mu = 500~\munit{m}$. The initial conditions are
\begin{equation}
[u_r,v_r,r,x,y,\psi]^T = [0,0,0,10,200,\pi/2]^T.
\end{equation}
The resulting motion of the ship can be seen in Figure \ref{COG-fig:Glo_pathzz}. The dashed blue line is the trajectory of the vessel and the red line is the reference. The yellow ship shows the orientation of the ship each $100~\munit{s}$. From Figure \ref{COG-fig:Glo_pathzz}, it can be seen that the ocean current prevents a tangential orientation of the ship with respect to the desired path, which is expected in order to compensate for the ocean current. 
\begin{figure}[tbh]
\centering
\includegraphics[width=0.5\columnwidth]{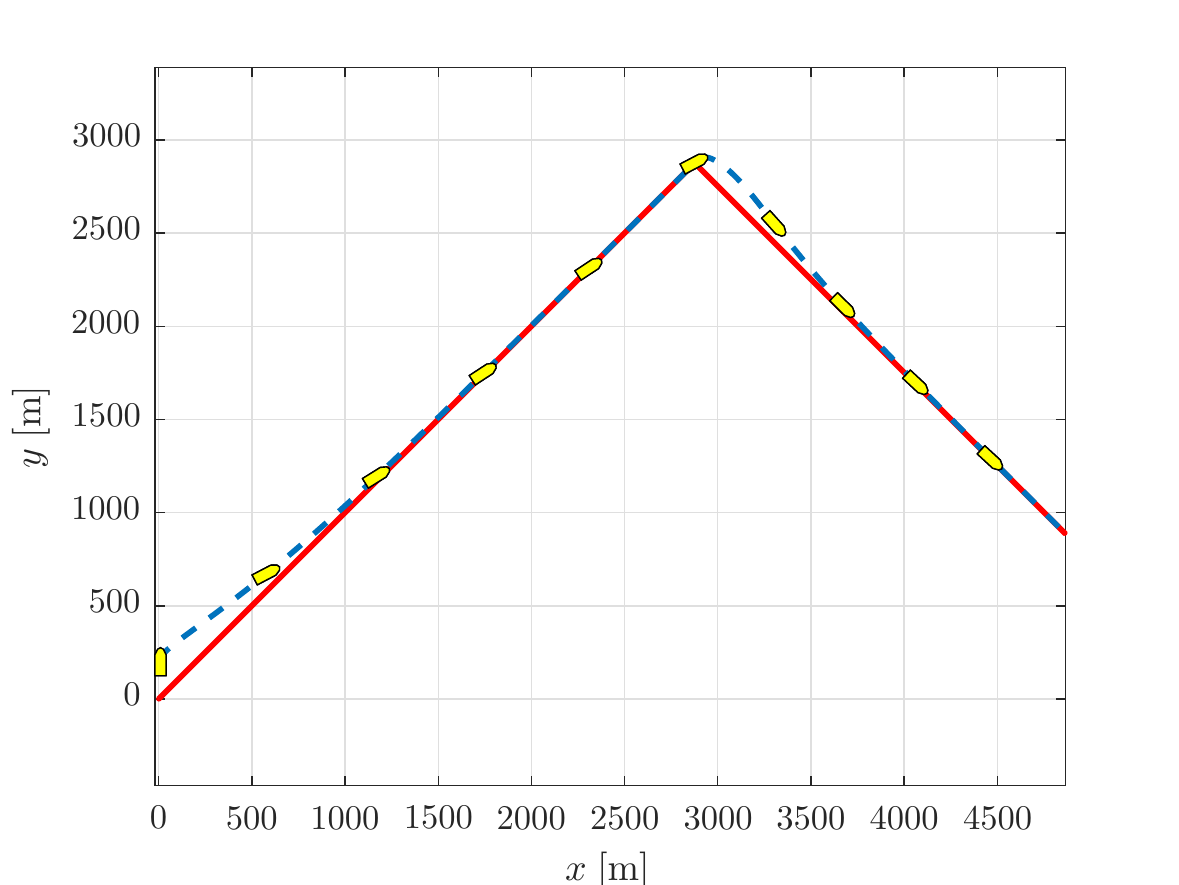}{}
\caption[Path following of the desired path.]{Path following of the desired zig-zag path in the $x-y$-plane using the proposed controller.} \label{COG-fig:Glo_pathzz}
\end{figure}
Furthermore, the estimates for the ocean current components obtained from the ocean current observer are given in the second plot from the top in Figure \ref{COG-fig:Glo_sub_cpzz}. From this plot, it can clearly be seen that the estimates converge to the desired value $(V_{x1}, V_{y1})$. Moreover, when the ocean current switches to $(V_{x2}, V_{y2})$, the observer shows a short transient behavior before converging to the new value. However, we preserve the conservativeness of the bound $2V_{\max} < u_{rd}(t),~\forall t$ derived in the analysis of the observer error dynamics in Subsection \ref{COG-subsec:obs}. Furthermore, during the transient interval during which the observer provides incorrect current estimates to the controller and also due to the switch of the desired-path direction, the path-following errors in the tangential and the normal directions $x_{b/p}$ and $y_{b/p}$, respectively, are affected, as it can be seen from the top plot of Figure \ref{COG-fig:Glo_sub_cpzz}. However, it can clearly be seen that the path-following errors converge back to zero after the transient. A detail of the last portions of the simulation is given to illustrate that the errors converge to zero. The sway velocity $v_r$ is plotted in the third plot of Figure \ref{COG-fig:Glo_sub_cpzz}. This plot shows that due to the fact that the curvature is zero almost everywhere on the path, the sway velocity converges to zero while the ship moves along the first segment of the path. When the ocean current and the path direction switch, the sway velocity displays a short transient behavior  before converging to the origin again. The relative surge velocity is plotted in the fourth plot of Figure \ref{COG-fig:Glo_sub_cpzz} and shows the exponential convergence of the velocity to $u_{rd} = 5~\munit{m/s}$. Furthermore, the evolution of $C_r$, involved in Condition \ref{COG-cond:Cr}, is plotted in the bottom plot of Figure \ref{COG-fig:Glo_sub_cpzz} where we see that $C_r$ is bounded away from zero throughout the motion. 
\begin{figure}[tbh]
\centering
\includegraphics[width=0.5\columnwidth]{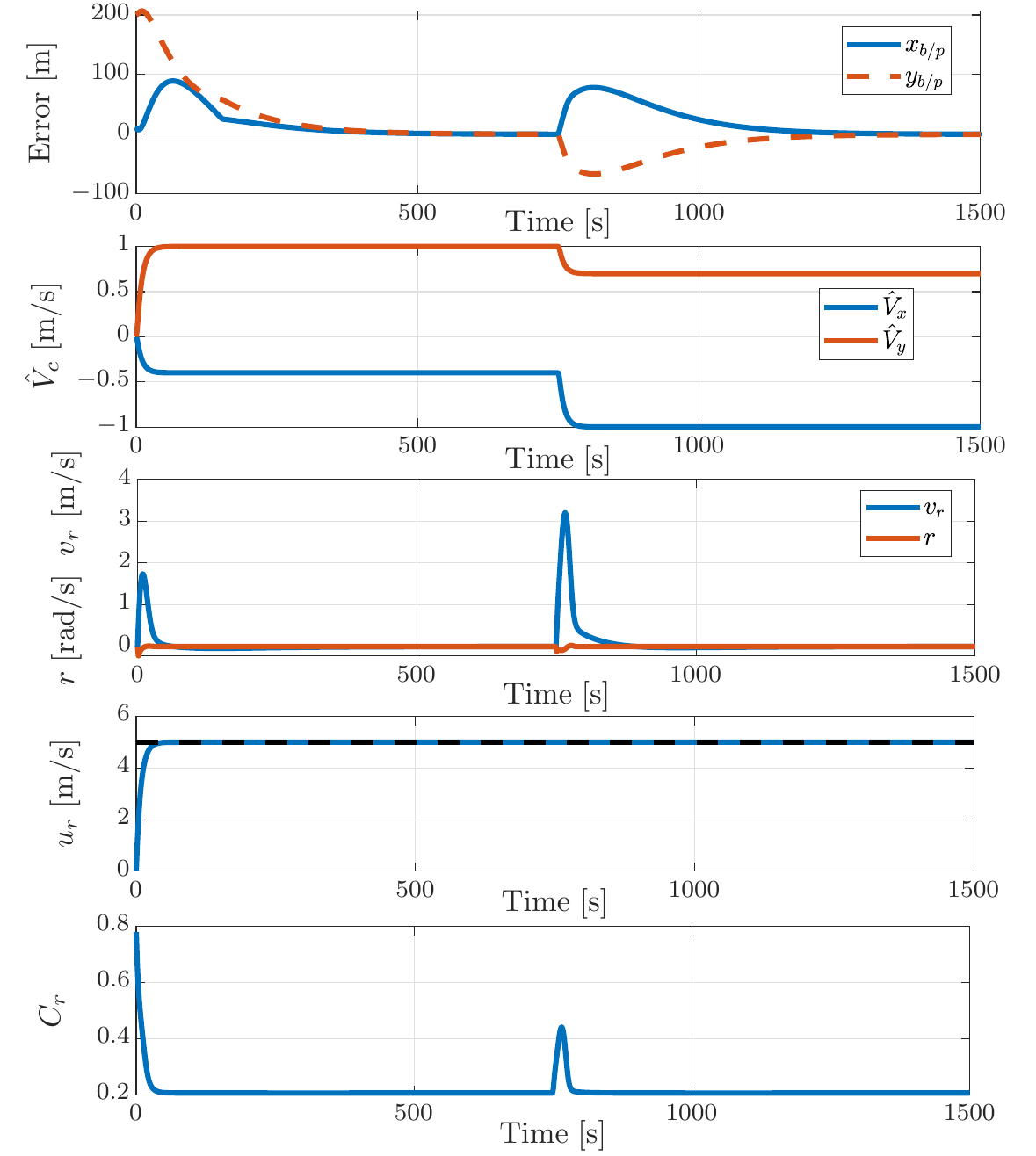}{}
\caption[The time evolution of the path-following errros, current estimates, sway velocity, surge velocity, and the parameter $C_r$.]{Time evolution of the path following errors (top), current estimates (second), sway velocity, yaw rate (third), surge velocity (fourth), and magnitude of $C_r$ (bottom).} \label{COG-fig:Glo_sub_cpzz}
\end{figure}

\section{Conclusions} \label{COG-sec:cncl}
In this work curved-path following for underactuated marine vehicles in the presence of constant ocean currents has been considered. We propose a control approach where the path is parametrised by a path variable with a globally defined update law. The vehicle is steered using a line-of-sight like guidance law where the lookahead-distance depends on the path-following errors. To compensate for the unknown ocean currents, the guidance law is aided by an ocean current observer. The closed-loop system with the controllers and observer was analysed. This was done by first showing that under certain conditions we have boundedness of the underactuated sway velocity dynamics. Since the paths are curved, the sway velocity will not converge to zero, and boundedness is thus what we aim for. It was then shown that if these conditions are satisfied such that the sway velocity is bounded, the path-following errors are globally asymptotically stable and locally exponentially stable.

\newpage

\bibliography{biblio}
\bibliographystyle{plainnat}

\appendix

\section*{Appendix}

\subsection*{Proof of Lemma \ref{COG-lem1}}

Consider the following part of the global closed-loop system:
\begin{small}
\begin{subequations} \label{COG-eq:36}
\begin{align} 
\begin{split}
\begin{bmatrix} \dot{\tilde{\psi}} \\  \dot{\tilde{r}} \end{bmatrix} = & \begin{bmatrix}   C_r \tilde{r} \\ - k_1 \tilde{r} - k_2C_r \tilde{\psi} \end{bmatrix} \\ &\hspace{-1cm}+ 
\underbrace{\begin{bmatrix}  G_2(\Delta,y_{b/p},x_{b/p},g,\hat{V}_N,\hat{V}_T,\tilde{V}_N,\tilde{V}_T)  \\ 
 - \tfrac{\partial r_d}{ \partial \tilde{\psi}} G_2(\cdot) - \tfrac{\partial r_d}{\partial y_{b/p}} \tilde{V}_{N} - \tfrac{\partial r_d}{\partial x_{b/p}} \tilde{V}_{T} + \tfrac{\partial r_d}{\partial \tilde{x}} \tilde{V}_x + \tfrac{\partial r_d}{ \partial \tilde{y}} \tilde{V}_y   \end{bmatrix} }_{R(h, y_{b/p}, \delta_x, \tilde{\psi}, \tilde{x}, \tilde{y})} \end{split} \label{COG-eq:36a} \\
\dot{v}_r = &  X ( u_{rd} + \tilde{u} ) r_d ( \cdot) + X( u_{rd} + \tilde{u} ) \tilde{r} + Y( u_{rd} + \tilde{u} ) v_r.  \label{COG-eq:36b}
\end{align}
\end{subequations}
\end{small}
From the boundedness of the vector $ [ \tilde{X}^T_2 , \kappa(\theta), u_{rd} , \dot{u}_{rd}, V_T, V_N ]^T $ we know that  $ \left\| [ \tilde{X}^T_2 , \kappa(\theta), u_{rd} , \dot{u}_{rd}, V_T, V_N ]^T \right\| \leq \beta_0 $, and from \eqref{COG-eq:rdgl} we can conclude the existence of positive functions $ a_{r_d}(\cdot) $, $ b_{r_d}(\cdot) $, $ a_{R}(\cdot) $, and $ b_R(\cdot) $ which are all continuous in their arguments and are such that such the following inequalities hold:
\begin{align} \label{eq:born1}
 \left| r_d(\cdot) \right| \leq &~a_{r_d}(\mu, \beta_0) \left| v_r \right|+ b_{r_d} (\mu, \beta_0)	
\end{align}
and,
\begin{align} \label{COG-eq:born2}
\left\| R(\cdot) \right\|	\leq a_R( \mu, \beta_0 ) \left| v_r \right| + b_R ( \mu, \beta_0 )
\end{align}
Then we choose the following Lyapunov function candidate:
\begin{align} \label{COG-eq:V}
V_1(\tilde{\psi}, \tilde{r}, v_r) = \tfrac{1}{2} \left( k_2\tilde{\psi}^2 + \tilde{r}^2 + v^2_r  \right)
\end{align}
whose time derivative along the solutions of \eqref{COG-eq:36} is
\begin{align} \label{COG-eq:V1}
\begin{split}
\dot{V}_1(\cdot) = &~k_2C_r \tilde{r} \tilde{\psi} - k_1 \tilde{r}^2 - k_2C_r \tilde{r} \tilde{\psi} + [\tilde{\psi}~~\tilde{r}] R(\cdot) \\ &+ Y(u_{rd}+ \tilde{u}) v^2_r + X(u_{rd} + \tilde{u}) \tilde{r} v_r \\& + X(u_{rd}+\tilde{u}) r_d(\cdot) v_r  
\end{split}
\end{align}
Using Young's inequality we note that
\begin{align} \label{COG-eq:dVfc}
\begin{split}
\dot{V}_1(\cdot)  \leq &~k_1 \tilde{r}^2 + \tilde{\psi}^2 + \tilde{r}^2 + R^2(\cdot) + Y(u_{rd} + \tilde{u}) v^2_r \\ &+ \left| X( u_{rd} + \beta_0 ) \right|  \left( \tilde{r}^2 + v^2_r \right) \\ &  +\left| X (u_{rd} + \beta_0) \right| \left( r^2_d(\cdot) + v^2_r \right)  \\ \leq &~ \alpha V_1 + \beta, ~~\alpha \geq 0,~~\beta \geq 0
\end{split}
\end{align}
Note that since the differential inequality \eqref{COG-eq:dVfc} is scalar we can invoke the comparison lemma \citet[Lemma 3.4]{khalil2002nonlinear}. From {\citet[Lemma 3.4]{khalil2002nonlinear}} we know that the solutions of differential inequality \eqref{COG-eq:dVfc} are bounded by the solutions of the linear system:
\begin{equation}
\dot{x} = \alpha x + \beta 
\end{equation}
which has solutions
\begin{equation}
x(t) = \tfrac{\|x(t_0)\|\alpha + \beta}{\alpha}e^{\alpha(t-t_0)} - \tfrac{\beta}{\alpha}
\end{equation}
Hence, from {\citet[Lemma 3.4]{khalil2002nonlinear}} we have that 
\begin{equation}
V_1(\cdot) \leq \tfrac{\|V_1(t_0)\|\alpha + \beta}{\alpha}e^{\alpha(t-t_0)} - \tfrac{\beta}{\alpha}
\end{equation}
which shows the solutions of $V_1(\cdot)$ are defined up to $t_{\max} = \infty$ and consequently from \eqref{COG-eq:V} it follows that the solutions of $\tilde{\psi}$, $\tilde{r}$, and $v_r$ must be defined up to $t_{\max} = \infty$. Hence, the solutions of \eqref{COG-eq:36} satisfy the definition of forward completeness presented in \citet{angeli1999forward} and we can conclude forward completeness of trajectories of \eqref{COG-eq:36}. 

The forward completeness of trajectories of the global closed-loop system now depends on forward completeness of of $\dot{y}_{b/p}$ and $\dot{x}_{b/p}$ from \eqref{COG-eq:35a}. We can conclude forward completeness of $\dot{y}_{b/p}$ and $\dot{x}_{b/p}$ by considering the Lyapunov function
\begin{equation} \label{COG-eq:Vfc2}
V_{2} = \tfrac{1}{2} x^2_{b/p} + \tfrac{1}{2} y^2_{b/p}.
\end{equation}
The time derivative of \eqref{COG-eq:Vfc2} is given by
\begin{align}
\begin{split}
\dot{V}_{2} &= x_{b/p}\dot{x}_{b/p} + y_{b/p}\dot{y}_{b/p} \\
&\leq -u_{td} \tfrac{y^2_{b/p}}{\sqrt{\Delta^2+(y_{b/p}+g)^2}} - \tfrac{ k_\delta x^2_{b/p}}{\sqrt{1+x^2_{b/p}}} \\
& + (G_1(\cdot)+\tilde{V}_N)y_{b/p} + \tilde{V}_T x_{b/p} \\
&\leq (G_1 + \tilde{V}_N)y_{b/p} + \tilde{V}_T x_{b/p}
\end{split}
\end{align}
where using the bound on $G_1(\cdot)$ from \eqref{COG-eq:Gbnd} and Young's inequality we obtain
\begin{align} \label{COG-eq:dVfc2}
\dot{V}_2 &\leq V_2 + \tfrac{1}{2}\left(\zeta^2(u_{td})\Vert [\tilde{\psi},\tilde{r}]^T \Vert^2 + \tilde{V}^2_N + \tilde{V}^2_T\right) \\
&\leq V_2 + \sigma_2(v_r,\tilde{\psi},\tilde{r},\tilde{V}_N,\tilde{V}_T)
\end{align}
with $\sigma_2(\cdot) \in \mc{K}_{\infty}$. Consequently, if we view the arguments of $\sigma_2(\cdot)$ as input to the $x_{b/p}$ and $y_{b/p}$ dynamics, then \eqref{COG-eq:dVfc2} satisfies {\citet[Corollary 2.11]{angeli1999forward}} and hence $\dot{x}_{b/p}$ and $\dot{y}_{b/p}$ are forward complete. Note that the arguments of $\sigma_2(\cdot)$ are all forward complete and therefore fit the definition of an input signal given in \citet{angeli1999forward}. We have now shown forward completeness of \eqref{COG-eq:35a} and \eqref{COG-eq:35c} and since \eqref{COG-eq:35b} is GES is is trivially forward complete. We can therefore claim forward completeness of the entire closed-loop system \eqref{COG-eq:35} and the proof of Lemma \ref{COG-lem1} is complete. \hfill $\blacksquare$

\subsection*{Proof of Lemma \ref{COG-lem2}} 

Recall the sway velocity dynamics \eqref{COG-eq:35c}:
\[  \dot{v}_r = X(\tilde{u} + u_{rd}) (r_d + \tilde{r}) + Y(u_{rd}+\tilde{u}) v_r ,~~Y(u_{rd}) < 0  \]
Consider the following Lyapunov function candidate:
\begin{align} \label{COG-eq:lyap}
	V_3(v_r) = \tfrac{1}{2} v^2_r
\end{align}
The derivative of \eqref{COG-eq:lyap} along the solutions of \eqref{COG-eq:35c} is given by
\begin{align} \label{COG-eq:dlyapLem2}
\begin{split}
\dot{V}_3 = &~v_r \dot{v}_r = v_r X(u_{rd} + \tilde{u}) r_d + X(u_{rd} + \tilde{u}) v_r \tilde{r} \\
&\qquad \quad + Y(u_{rd} + \tilde{u}) v^2_r \\ 
\leq &~X(u_{rd}) r_d v_r + a_x \tilde{u} r_d v_r + X(u_{rd}) v_r \tilde{r} \\
& + a_x \tilde{u} v_r \tilde{r} + a_y \tilde{u} v^2_r + Y(u_{rd}) v^2_r  
\end{split}
\end{align}
where we used the fact that: 
\begin{align} \label{COG-eq:YX}
 Y(u_{r}) = &~a_y u_r + b_y \\
 X(u_r) = &~a_x u_r + b_x    
\end{align}
The term $r_dv_r$ is given by
\begin{align} \label{COG-eq:rdvrLem2full}
\begin{split}
r_d v_r = &-\tfrac{v_r}{C_r} \left[ \kappa(\theta) \left( u_t \cos(\psi + \beta - \gamma_p(\theta)) + k_\delta \tfrac{x_{b/p}}{\sqrt{ 1 + x^2_{b/p} }} +\hat{V}_T \right) \right. \\ &+ \tfrac{ Y(u_r) v_r  u_{rd} - \dot{u}_{rd} v_r } { u^2_{rd} + v^2_r } + \tfrac{ \Delta }{\Delta^2 + \left( y_{b/p} + g \right)^2 } \left[  \dot{\hat{V}}_{N} \tfrac{ b + \sqrt{b^2-ac} }{-a}  \right.\\ 
&+  \tfrac{\partial g}{\partial b} \left( 2 \dot{\hat{V}}_{N} y_{b/p} \right)+\tfrac{\partial g}{\partial a} \left( 2 \hat{V}_{N} \dot{\hat{V}}_{N} - 2 u_{rd} \dot{u}_{rd} - 2 v_r Y(u_r) v_r \right)\\ 
&+ \left[ 1 + \tfrac{\partial g}{\partial c} 2 y_{b/p} + \tfrac{\partial g}{\partial b} 2 \hat{V}_{N} \right] \left( \tfrac{- u_{td} y_{b/p} }{ \sqrt{\Delta^2 + (y_{b/p} + g)^2} } + G_1(\cdot) - x_{b/p} \kappa(\theta)\dot{\theta} \right) \\ 
&+ \left. \tfrac{\partial g}{\partial c} 2 \Delta \left[ \tfrac{\partial \Delta}{\partial x_{b/p}}  \left( - k_\delta \tfrac{x_{b/p}}{\sqrt{ 1 + x^2_{b/p} }} + y_{b/p}\kappa(\theta)\dot{\theta} \right) \right.\right. \\ 
&+ \left.\left. \tfrac{\partial \Delta}{\partial y_{b/p}} \left( - u_{td}   \tfrac{ y_{b/p} }{ \sqrt{\Delta^2 + (y_{b/p} + g)^2} } + G_1(\cdot) - x_{b/p} \kappa(\theta)\dot{\theta} \right) \right] \right] \\
&- \tfrac{y_{b/p} + g}{\Delta^2 + (y_{b/p} + g)^2} \left[ \tfrac{\partial \Delta}{\partial x_{b/p}}  \left( - k_\delta \tfrac{x_{b/p}}{\sqrt{ 1 + x^2_{b/p} }} + y_{b/p}\kappa(\theta)\dot{\theta} \right) \right. \\  
&+ \left.\left. \tfrac{\partial \Delta}{\partial y_{b/p}} \left( - u_{td}  \tfrac{ y_{b/p} }{ \sqrt{\Delta^2 + (y_{b/p} + g)^2} } + G_1(\cdot) - x_{b/p} \kappa(\theta)\dot{\theta}\right) \right] \right] \\ 
\end{split}
\end{align}
We now introduce a term $F(\tilde{X}_1, \tilde{X}_2, \Delta, V_{T}, V_{T}, u_{rd}, v_r)$ to collect all the terms that grow linearly with $v_r$ and the terms that grow quadratically with $v_r$ but vanish when $\tilde{X}_1$ and $\tilde{X}_2$ are zero. Consequently we rewrite \eqref{COG-eq:rdvrLem2full} to obtain
\begin{align} \label{COG-eq:rdvrLem2full2}
\begin{split}
r_dv_r = &-\tfrac{v_r}{C_r} \left[1 + \tfrac{ \Delta x_{b/p} }{ \Delta^2 + \left( y_{b/p} + g \right)^2 } \right] \kappa(\theta) \left( u_t \cos(\psi + \beta - \gamma_p(\theta))  \right)\\ 
&- \tfrac{1}{C_r} \left( \tfrac{ u_{rd}}{ u^2_{rd} + v^2_r } - \tfrac{2\Delta v_r}{\Delta^2 + (y_{b/p} + g)^2} \tfrac{\partial g }{\partial a}\right)Y(u_r)v^2_r \\
&+ F(\tilde{X}_1, \tilde{X}_2, \Delta, V_{T}, V_{T}, u_{rd}, v_r) 
\end{split}
\end{align}
where
\begin{align}
\begin{split}
F(\cdot) = &-\tfrac{v_r}{C_r} \left[ \kappa(\theta) \left( k_\delta \tfrac{x_{b/p}}{\sqrt{ 1 + x^2_{b/p} }} +\hat{V}_T \right) - \tfrac{\dot{u}_{rd} v_r } { u^2_{rd} + v^2_r } \right. \\ & + \tfrac{ \Delta }{\Delta^2 + \left( y_{b/p} + g \right)^2 } \left[  \dot{\hat{V}}_{N} \tfrac{ b + \sqrt{b^2-ac} }{-a} +  \tfrac{\partial g}{\partial b} \left( 2 \dot{\hat{V}}_{N} y_{b/p} \right) \right.\\ 
&+2\frac{\partial g}{\partial a} \left( \hat{V}_{N} \dot{\hat{V}}_{N} - u_{rd} \dot{u}_{rd} \right)- x_{b/p}\kappa(\theta)\left(\tfrac{k_\delta x_{b/p}}{\sqrt{ 1 + x^2_{b/p} }} + \hat{V}_T \right)\\ 
&+ \left[ \tfrac{\partial g}{\partial c} 2 y_{b/p} + \tfrac{\partial g}{\partial b} \left( 2 \hat{V}_{N} \right) \right] \left( \tfrac{- u_{td} y_{b/p} }{ \sqrt{\Delta^2 + (y_{b/p} + g)^2} } + G_1(\cdot) - x_{b/p} \kappa(\theta)\dot{\theta} \right) \\ 
&- \left. \tfrac{\partial g}{\partial c} 2 \Delta \left[\tfrac{\partial \Delta}{\partial x_{b/p}}   \tfrac{k_\delta x_{b/p}}{\sqrt{ 1 + x^2_{b/p} }} + \tfrac{\partial \Delta}{\partial y_{b/p}} \left( \tfrac{ u_{td}   y_{b/p} }{ \sqrt{\Delta^2 + (y_{b/p} + g)^2} } + G_1(\cdot) \right) \right] \right] \\
&+ \tfrac{y_{b/p} + g}{\Delta^2 + (y_{b/p} + g)^2} \left[ \tfrac{\partial \Delta}{\partial x_{b/p}} \tfrac{k_\delta x_{b/p}}{\sqrt{ 1 + x^2_{b/p} }} \right. \\  &\hspace{3.5cm}+ \left.\left. \tfrac{\partial \Delta}{\partial y_{b/p}} \left(  \tfrac{ u_{td}  y_{b/p} }{ \sqrt{\Delta^2 + (y_{b/p} + g)^2} } + G_1(\cdot) \right)\right] \right]
\end{split}
\end{align}
Note here that using our definition of $\Delta$ in \eqref{COG-eq:Delta} all the terms in $r_dv_r$ with partial derivatives of $\Delta$ multiplied by $\dot{\theta}$ are cancelled due to skew-symmetry. It is straightforward to verify that the function $F(\cdot)$ satisfies the following inequality:
\begin{align} \label{COG-eq:F}
\left|	F(\cdot) \right| \leq F_2 (\tilde{X}_1, \tilde{X}_2, \Delta, V_{T}, V_{N}, u_{rd}) v^2_r + F_1(\tilde{X}_1, \tilde{X}_2, \Delta, V_{T}, V_{N}, u_{rd}) \left| v_r \right|
\end{align}
where $F_{1,2}(\cdot)$ are positive functions continuous in their arguments with:
\begin{align} \label{COG-eq:F2}
F_2( 0,0,\Delta, V_{T}, V_{N}, u_{rd} ) = 0.
\end{align}
Consequently, using \eqref{COG-eq:rdvrLem2full2} the term $r_d v_r$ can be bounded as a function of $v_r$ as follows
\begin{align} \label{COG-eq:rdvrLem2bound}
\begin{split}
r_d v_r \leq &~\sqrt{u^2_r + v^2_r}\left| \tfrac{v_r}{C_r} \right| \left| \kappa(\theta) \right| \left|\left[-1 + \tfrac{ \Delta x_{b/p} }{ \Delta^2 + \left( y_{b/p} + g \right)^2 } \right] \right|+ \left| F(\cdot) \right| \\ 
&- \tfrac{1}{C_r} \left( \tfrac{ u_{rd}}{ u^2_{rd} + v^2_r } - \tfrac{2\Delta v_r}{\Delta^2 + (y_{b/p} + g)^2} \tfrac{\partial g }{\partial a}\right)Y(u_r)v^2_r \\ 
\leq &~\left| \tfrac{v^2_r}{C_r} \right| \left| \kappa(\theta) \right| \left| \left[-1 + \tfrac{ \Delta x_{b/p} }{ \Delta^2 + \left( y_{b/p} + g \right)^2 } \right] \right|+ \left| F(\cdot) \right| \\ 
&+ \left| \tfrac{v_r}{C_r} \right| \left| \kappa(\theta) \right| \left|u_r\right|  \left| \left[-1 + \tfrac{ \Delta x_{b/p} }{ \Delta^2 + \left( y_{b/p} + g \right)^2 } \right] \right| \\ 
&- \tfrac{1}{C_r} \left( \tfrac{ u_{rd}}{ u^2_{rd} + v^2_r } - \tfrac{2\Delta v_r}{\Delta^2 + (y_{b/p} + g)^2} \tfrac{\partial g }{\partial a}\right)Y(u_r)v^2_r
\end{split}
\end{align}

\begin{remark} \label{COG-rem:Delta}
The necessity for the choice of $\Delta$ as in \eqref{COG-eq:Delta} becomes evident from \eqref{COG-eq:rdvrLem2full2}. The choice of $\Delta$ constant would make all partial derivatives of $\Delta$ equal to zero. However, from $v_r/C_rx_{b/p}\kappa(\theta)\dot{\theta}$ we obtain a term of the form
\begin{equation} \label{COG-eq:RemLem21}
\tfrac{v^2_r}{C_r}\kappa(\theta)\tfrac{\Delta^2 x_{b/p}}{(\Delta^2 + (y_{b/f}+g)^2)^{3/2}}
\end{equation} 
which grows quadratically in $v_r$ with a gain that cannot be bounded independent of $x_{b/f}$ if $\Delta$ is independent of $x_{b/f}$. Therefore, boundedness of $v_r$ cannot be shown independently of $x_{b/f}$. With the choice of $\Delta = \sqrt{\mu^2+x^2_{b/p}}$ as proposed in \citet{moe2014path}, the partial derivatives with respect to $y_{b/p}$ would be zero. The term in \eqref{COG-eq:RemLem21} would now be upper-bounded by one. However, a new term would then be introduced from the partial derivative of $\Delta$
\begin{equation}
\tfrac{v^2_r}{C_r}\tfrac{\partial \Delta}{\partial x_{b/p}}\kappa(\theta)\tfrac{\Delta y_{b/p}(y_{b/p}+g)}{(\Delta^2 + (y_{b/f}+g)^2)^{3/2}}
\end{equation} 
where it should be noted that this term can grow unbounded in $y_{b/p}$ near the manifold where $g=-(y_{b/p}+1)$. Hence, the growth of this quadratic term in $v_r$ cannot be upper-bounded independent of $y_{b/p}$. 
\end{remark}

To avoid the issues describe in Remark \ref{COG-rem:Delta}, we choose $\Delta$ as defined in \eqref{COG-eq:Delta}. Using the definition of $\Delta(x_{b/p},y_{b/p})$ given in \eqref{COG-eq:Delta} it is straightforward to verify that
\begin{align} \label{COG-eq:rdvrLem2}
\begin{split}
r_d v_r \leq &~\left| \tfrac{v^2_r}{C_r} \right| \left| \kappa(\theta) \right| \left| \left[-1 + \tfrac{ \Delta x_{b/p} }{ \Delta^2 + \left( y_{b/p} + g \right)^2 } \right] \right|+ \left| F(\cdot) \right| \\  &+ \left| \tfrac{v_r}{C_r} \right| \left| \kappa(\theta) \right| \left|u_r\right|  \left| \left[-1 + \tfrac{ \Delta x_{b/p} }{ \Delta^2 + \left( y_{b/p} + g \right)^2 } \right] \right| \\ 
&- \tfrac{1}{C_r} \left( \tfrac{ u_{rd}}{ u^2_{rd} + v^2_r } - \tfrac{2\Delta v_r}{\Delta^2 + (y_{b/p} + g)^2} \tfrac{\partial g }{\partial a}\right)Y(u_r)v^2_r \\
\leq &~2\left| \tfrac{v^2_r}{C_r} \right| \left| \kappa(\theta) \right| + 2\left|u_r\right|\left| \tfrac{v_r}{C_r} \right| \left| \kappa(\theta) \right| + \left| F(\cdot) \right| \\ 
&- \tfrac{1}{C_r} \left( \tfrac{ u_{rd}}{ u^2_{rd} + v^2_r } - \tfrac{2\Delta v_r}{\Delta^2 + (y_{b/p} + g)^2} \tfrac{\partial g }{\partial a}\right)Y(u_r)v^2_r
\end{split}
\end{align}
When substituting \eqref{COG-eq:rdvrLem2} in \eqref{COG-eq:dlyapLem2} we obtain
\begin{align} \label{COG-eq:dotvr}
\begin{split}
\dot{V}_3 = v_r \dot{v}_r \leq &~\tfrac{1}{C_r}\Big[ 2\left| X(u_{rd}) \right| \left| \kappa(\theta) \right| + Y(u_{rd}) \Big] v^2_r + a_y \tilde{u} v^2_r + a_x \tilde{u} v_r \tilde{r}  
\\ &+X(u_{rd}) \left(F(\cdot) + 2\left|u_r\right|\left| \tfrac{v_r}{C_r} \right|\right) + a_x \tilde{u} r_d v_r + X(u_{rd}) v_r \tilde{r} 
\end{split}   
\end{align}

Consequently, on the manifold where $(\tilde{X}_1,\tilde{X}_2)=0$ we have
\begin{align} \label{COG-eq:Vdotmanlem2}
\begin{split}
\dot{V}_3 &\leq \tfrac{1}{C^*_r} \Big(2X_{\max}\left| \kappa(\theta) \right| + Y_{\min}\Big)v^2_r + X(u_{rd}) F_1 ( 0,0, \Delta, V_T, V_N, u_{rd})|v_r|
\end{split}
\end{align}
where $C^*_r (v_r, x_{b/p}, y_{b/p}, \Delta , V_N, u_{rd} ) = C_r( v_r,x_{b/p}, y_{b/p}, \Delta, \hat{V}_N = V_N , u_r = u_{rd} )$. Boundedness of \eqref{COG-eq:Vdotmanlem2} is guaranteed as long as 
\begin{align} \label{COG-eq:lem2c1}
2X_{\max} \left| \kappa(\theta) \right| + Y_{\min} < 0
\end{align} 
Hence, satisfaction of \eqref{COG-eq:curvlem2} renders the quadratic term in \eqref{COG-eq:Vdotmanlem2} negative and since the quadratic term is dominant for sufficiently large $v_r$, \eqref{COG-eq:Vdotmanlem2} is negative definite for sufficiently large $v_r$. If $\dot{V}_3$ is negative for sufficiently large $v_r$ this implies that $V_3$ decreases for sufficiently large $v_r$. Since $V_3 = 1/2v^2_r$, a decrease in $V_3$ implies a decrease in $v^2_r$ and by extension in $v_r$. Therefore, $v_r$ cannot increase above a certain value and $v_r$ is bounded near the manifold where $(\tilde{X}_1,\tilde{X}_2)=0$.  

\begin{remark}
Note that $C^*_r (v_r, y_{b/f} , \Delta , V_{N}, u_{rd} )$ can be found independently of $y_{b/p}$ and $x_{b/p}$ since the terms in $C_r$ are bounded with respect to these variables. 
\end{remark}

Consequently, close to the manifold where $(\tilde{X}_1,\tilde{X}_2)=0$ the sufficient and necessary condition for local boundedness of \eqref{COG-eq:35c} is the following:
\begin{align}
 2X_{\max} \left| \kappa(\theta) \right| + Y_{\min} < 0
\end{align} 
which is satisfied if and only if the condition in Lemma \ref{COG-lem2} is satisfied. This completes the proof of Lemma \ref{COG-lem2}. \hfill $\blacksquare$

\subsection*{Proof of Lemma \ref{COG-lem3}} 

Recall the sway velocity dynamics \eqref{COG-eq:35c}:
\[  \dot{v}_r = X(\tilde{u} + u_{rd}) (r_d + \tilde{r}) + Y(u_{rd}+\tilde{u}) v_r ,~~Y(u_{rd}) < 0  \]
Consider the following Lyapunov function candidate:
\begin{align} \label{COG-eq:lyap2}
	V_3(v_r) = \tfrac{1}{2} v^2_r
\end{align} 
The derivative of \eqref{COG-eq:lyap2} along the solutions of \eqref{COG-eq:35c} is given by
\begin{align} \label{COG-eq:dlyap2}
\begin{split}
\dot{V}_3 = &~v_r \dot{v}_r = v_r X(u_{rd} + \tilde{u}) r_d + X(u_{rd} + \tilde{u}) v_r \tilde{r} + Y(u_{rd} + \tilde{u}) v^2_r \\ \leq &~X(u_{rd}) r_d v_r + a_x \tilde{u} r_d v_r + X(u_{rd}) v_r \tilde{r} + a_x \tilde{u} v_r \tilde{r} + a_y \tilde{u} v^2_r + Y(u_{rd}) v^2_r  
\end{split}
\end{align}
where we used the fact that: 
\begin{align} \label{COG-eq:YX2}
 Y(u_{r}) = &~a_y u_r + b_y \\
 X(u_r) = &~a_x u_r + b_x    
\end{align}
The term $r_d v_r$ is given by:
\begin{align} \label{COG-eq:rdvrLem3full}
\begin{split}
r_d v_r = &-\tfrac{v_r}{C_r} \left[ \kappa(\theta) \left( u_t \cos(\psi + \beta - \gamma_p(\theta)) + k_\delta \tfrac{x_{b/p}}{\sqrt{ 1 + x^2_{b/p} }} +\hat{V}_T \right) \right. \\ &+ \tfrac{ Y(u_r) v_r  u_{rd} - \dot{u}_{rd} v_r } { u^2_{rd} + v^2_r } + \tfrac{ \Delta }{\Delta^2 + \left( y_{b/p} + g \right)^2 } \left[  \dot{\hat{V}}_{N} \tfrac{ b + \sqrt{b^2-ac} }{-a}  \right.\\ 
&+  \tfrac{\partial g}{\partial b} \left( 2 \dot{\hat{V}}_{N} y_{b/p} \right) 
+\tfrac{\partial g}{\partial a} \left( 2 \hat{V}_{N} \dot{\hat{V}}_{N} - 2 u_{rd} \dot{u}_{rd} - 2 v_r Y(u_r) v_r \right)\\ 
&+ \left[ 1 + \tfrac{\partial g}{\partial c} 2 y_{b/p} + \tfrac{\partial g}{\partial b} 2 \hat{V}_{N} \right] \left(    \tfrac{- u_{td} y_{b/p} }{ \sqrt{\Delta^2 + (y_{b/p} + g)^2} } + G_1(\cdot) - x_{b/p} \kappa(\theta)\dot{\theta} \right) \\ 
&+ \left. \tfrac{\partial g}{\partial c} 2 \Delta \left[ \tfrac{\partial \Delta}{\partial x_{b/p}}  \left( - k_\delta \tfrac{x_{b/p}}{\sqrt{ 1 + x^2_{b/p} }} + y_{b/p}\kappa(\theta)\dot{\theta} \right) \right.\right. \\ 
&+ \left.\left. \tfrac{\partial \Delta}{\partial y_{b/p}} \left( - u_{td}   \tfrac{ y_{b/p} }{ \sqrt{\Delta^2 + (y_{b/p} + g)^2} } + G_1(\cdot) - x_{b/p} \kappa(\theta)\dot{\theta} \right) \right] \right] \\
&- \tfrac{y_{b/p} + g}{\Delta^2 + (y_{b/p} + g)^2} \left[ \tfrac{\partial \Delta}{\partial x_{b/p}}  \left( - k_\delta \tfrac{x_{b/p}}{\sqrt{ 1 + x^2_{b/p} }} + y_{b/p}\kappa(\theta)\dot{\theta} \right) \right. \\  
&+ \left.\left. \tfrac{\partial \Delta}{\partial y_{b/p}} \left( - u_{td}  \tfrac{ y_{b/p} }{ \sqrt{\Delta^2 + (y_{b/p} + g)^2} } + G_1(\cdot) - x_{b/p} \kappa(\theta)\dot{\theta}\right) \right]\right] \\ 
\end{split}
\end{align}
We can now collect the terms that have less than quadratic growth in $v_r$ and/or vanish when $\tilde{X}_2 = 0$.
\begin{align}
\begin{split}
r_dv_r= & -\tfrac{v_r}{C_r} \kappa(\theta) \left( \sqrt{u^2_r + v^2_r} \cos(\psi + \beta - \gamma_p(\theta))  \right) 
\\ &+ \tfrac{v_r}{C_r} \tfrac{ \Delta x_{b/p} }{ \Delta^2 + \left( y_{b/p} + g \right)^2 } \left( \kappa(\theta)  \sqrt{u^2_r+v^2_r} \cos(\psi+\beta-\gamma_p)  \right) 
\\ &- \tfrac{v_r}{C_r}  \tfrac{ \Delta }{\Delta^2 + \left( y_{b/p} + g \right)^2 } \left( - u_{td}  \tfrac{ y_{b/p} } { \sqrt{\Delta^2 + (y_{b/p} + g)^2 } } + G_1(\cdot)  \right) 
\\ &+ \tfrac{v_r}{C_r} \tfrac{ y_{b/p} + g }{ \Delta^2 + (y_{b/p}+g)^2 } \tfrac{\partial \Delta}{\partial y_{b/p}}  \left( - u_{td}  \tfrac{y_{b/p}}{\sqrt{\Delta^2 + (y_{b/p} + g)^2}} + G_1( \cdot ) \right) \\ 
&- \tfrac{1}{C_r} \left( \tfrac{ u_{rd}}{ u^2_{rd} + v^2_r } - \tfrac{2\Delta v_r}{\Delta^2 + (y_{b/p} + g)^2} \tfrac{\partial g }{\partial a}\right)Y(u_r)v^2_r
\\ &+ G(\tilde{X}_1, \tilde{X}_2, \Delta, V_{T}, V_{N}, u_{rd}, v_r) 
\end{split}
\end{align}
where,
\begin{align}
\begin{split}
G(\cdot)\triangleq &-\tfrac{v_r}{C_r} \left[ \kappa(\theta) \left( k_\delta \tfrac{x_{b/p}}{\sqrt{ 1 + x^2_{b/p} }} +\hat{V}_T \right) -  \tfrac{k_\delta x_{b/p}}{\sqrt{ 1 + x^2_{b/p} }}  \right. 
\\ &- \tfrac{\dot{u}_{rd} v_r } { u^2_{rd} + v^2_r } - \tfrac{y_{b/p} + g}{\Delta^2 + (y_{b/p} + g)^2} \tfrac{\partial \Delta}{\partial x_{b/p}} \\
&+ \tfrac{ \Delta }{\Delta^2 + \left( y_{b/p} + g \right)^2 } \left[  \dot{\hat{V}}_{N} \tfrac{ b + \sqrt{b^2-ac} }{-a} +  \tfrac{\partial g}{\partial b} \left( 2 \dot{\hat{V}}_{N} y_{b/p} \right) \right.\\
&- \tfrac{\partial g}{\partial c} 2 \Delta \left[\tfrac{\partial \Delta}{\partial x_{b/p}} \tfrac{ k_\delta x_{b/p}}{\sqrt{ 1 + x^2_{b/p} }} + \tfrac{\partial \Delta}{\partial y_{b/p}} \left( \tfrac{u_{td} y_{b/p} }{ \sqrt{\Delta^2 + (y_{b/p} + g)^2} } + G_1(\cdot) \right) \right] \\ 
&+ 2\left[\tfrac{\partial g}{\partial c} y_{b/p} + \tfrac{\partial g}{\partial b} \hat{V}_{N} \right] \left( \tfrac{- u_{td} y_{b/p} }{ \sqrt{\Delta^2 + (y_{b/p} + g)^2} } + G_1(\cdot) - x_{b/p} \kappa(\theta)\dot{\theta} \right) \\ 
&+ \left.\left. 2\frac{\partial g}{\partial a} \left( \hat{V}_{N} \dot{\hat{V}}_{N} - u_{rd} \dot{u}_{rd} \right) - x_{b/p}\kappa(\theta)\left(\tfrac{k_\delta x_{b/p}}{\sqrt{ 1 + x^2_{b/p} }} + \hat{V}_T \right) \right] \right]
\end{split}
\end{align}
where $G(\cdot)$ is the function introduced to collect the terms that have less than quadratic growth in $v_r$ and/or vanish when $\tilde{X}_2 = 0$. Note here that using our definition of $\Delta$ in \eqref{COG-eq:Delta} all the terms in $r_dv_r$ with partial derivatives of $\Delta$ multiplied by $\dot{\theta}$ are canceled due to skew-symmetry. We can now find the following bound on \eqref{COG-eq:rdvrLem3full}
\begin{align} \label{COG-eq:rdvrboundlem3}
\begin{split}
r_dv_r \leq & \left| \tfrac{v_r}{C_r} \right| \left| \kappa(\theta) \right| \sqrt{u^2_r + v^2_r}  \left| \tfrac{ \Delta x_{b/p} }{ \Delta^2 + \left( y_{b/p} + g \right)^2 } - 1 \right| \\ &+ \left| \tfrac{v_r}{C_r} \right| \left| \tfrac{1}{\Delta} \right|  \left( 4 \sqrt{u^2_r + v^2_r} + \left| \tilde{u} \right| \right) \\
&+ \left| \tfrac{v_r}{C_r} \right| \left| \tfrac{y_{b/p} + g}{\Delta^2 + (y_{b/p}+g)^2} \right|  \left( 4 \sqrt{u^2_r + v^2_r} + \left| \tilde{u} \right| \right) + \left| G(\cdot) \right| \\ 
&- \tfrac{1}{C_r} \left( \tfrac{ u_{rd}}{ u^2_{rd} + v^2_r } - \tfrac{2\Delta v_r}{\Delta^2 + (y_{b/p} + g)^2} \tfrac{\partial g }{\partial a}\right)Y(u_r)v^2_r
 \\ \leq & \left| \tfrac{v^2_r}{C_r} \right| \left[ \left| \kappa(\theta) \right| \left| \tfrac{ \Delta x_{b/p} }{ \Delta^2 + \left( y_{b/p} + g \right)^2 } -1  \right| +  \tfrac{8}{\Delta} \right] + \left| G(\cdot) \right| \\ &+ \left| \tfrac{v_r}{C_r} \right| \left| \kappa(\theta) \right| \left|u_r\right|  \left| \tfrac{ \Delta x_{b/p} }{ \Delta^2 + \left( y_{b/p} + g \right)^2 } -1 \right| + \left| \tfrac{v_r}{C_r} \right| \left| \tfrac{2}{\Delta} \right|  \left( 4 \left|u_r\right| + \left| \tilde{u} \right| \right) \\ 
&- \tfrac{1}{C_r} \left( \tfrac{ u_{rd}}{ u^2_{rd} + v^2_r } - \tfrac{2\Delta v_r}{\Delta^2 + (y_{b/p} + g)^2} \tfrac{\partial g }{\partial a}\right)Y(u_r)v^2_r
\\ \leq & \left| \tfrac{v^2_r}{C_r} \right| \left[ \left| \kappa(\theta) \right| \left| \tfrac{ \Delta x_{b/p} }{ \Delta^2 + \left( y_{b/p} + g \right)^2 } -1 \right| +  \tfrac{8}{\Delta} \right] + \Phi(\cdot) \\ 
&- \tfrac{1}{C_r} \left( \tfrac{ u_{rd}}{ u^2_{rd} + v^2_r } - \tfrac{2\Delta v_r}{\Delta^2 + (y_{b/p} + g)^2} \tfrac{\partial g }{\partial a}\right)Y(u_r)v^2_r \\ 
\leq & \left| \tfrac{v^2_r}{C_r} \right| \left[ 2\left| \kappa(\theta) \right| +  \tfrac{8}{\Delta} \right] + \Phi(\cdot) \\ 
&- \tfrac{1}{C_r} \left( \tfrac{ u_{rd}}{ u^2_{rd} + v^2_r } - \tfrac{2\Delta v_r}{\Delta^2 + (y_{b/p} + g)^2} \tfrac{\partial g }{\partial a}\right)Y(u_r)v^2_r
\end{split}
\end{align}
where,
%%%%%%%%%%%%%%%%%%%%%%%%%%%%%%%%
\begin{align}
\Phi(\cdot) \triangleq &  \left| G(\cdot) \right| +  2\left| \tfrac{v_r}{C_r} \right| \left| \kappa(\theta) \right| \left|u_r\right| + 2 \left| \tfrac{v_r}{C_r} \right| \left| \tfrac{1}{\Delta} \right|  \left( 4 \left|u_r\right| + \left| \tilde{u}_r \right| \right)
\end{align}
The function $\Phi(\cdot)$ is introduced to collect the remaining terms that have less than quadratic growth in $v_r$ and/or vanish when $\tilde{X}_2 = 0$. Note also the terms in $G(\cdot)$ with partial derivatives of $g$ that appear to have quadratic growth. Although the overall terms appear to have quadratic growth, the partial derivatives of $g$ actually decrease for increasing $v_r$ giving the entire term less than quadratic growth. From the definitions of $\Phi(\cdot)$ and $G(\cdot)$ one can easily conclude the existence of three continuous positive functions $ F_{0,2} ( \tilde{X}_1, \tilde{X}_2 , u_{rd} , \dot{u}_{rd} , V_{T} , V_{N}, \Delta ) $ which are bounded under the boundedness of the vector $ [ \tilde{X}^T_2, u_{rd}, \dot{u}_{rd}, V_{T}, V_{N}, \Delta ]^T $, with 
\[ 
F_{2} ( \tilde{X}_1, \tilde{X}_2 = 0 , u_{rd} , \dot{u}_{rd} , V_{x_e} , V_{y_e}, \Delta ) = 0,
 \] 
such that:
\begin{align}
\Phi(\cdot) \leq &~F_2(\cdot) v^2_r + F_1(\cdot) v_r + F_0(\cdot).
\end{align}

When we substitute the bound on $r_dv_r$ from \eqref{COG-eq:rdvrboundlem3} in \eqref{COG-eq:dlyap2} we obtain:
\begin{align} \label{COG-eq:dlyaplem3}
\begin{split}
\dot{V}_3 =  v_r \dot{v}_r \leq & \left| X(u_{rd}) \right| \left(\left| \tfrac{v^2_r}{C_r} \right| \left[ 2\left| \kappa(\theta) \right| +  \tfrac{8}{\Delta} \right] + \Phi(\cdot)\right) + a_x \tilde{u} r_d v_r \\ & + X(u_{rd}) v_r \tilde{r} + a_x \tilde{u} v_r \tilde{r} + a_y \tilde{u} v^2_r + Y(u_{rd}) v^2_r \\ 
&- \tfrac{1}{C_r} \left( \tfrac{ u_{rd}}{ u^2_{rd} + v^2_r } - \tfrac{2\Delta v_r}{\Delta^2 + (y_{b/p} + g)^2} \tfrac{\partial g }{\partial a}\right)Y(u_r)v^2_r \\
\leq & \left| \tfrac{1}{C_r} \right| \left[ \left| X(u_{rd}) \right| \left[ 2\left| \kappa(\theta) \right| +  \tfrac{8}{\Delta} \right] - \left|Y(u_{rd})\right| \right]v^2_r \\
&+ a_x \tilde{u} r_d v_r + X(u_{rd})(v_r\tilde{r}+\Phi(\cdot)) + a_x \tilde{u} v_r \tilde{r} + a_y \tilde{u} v^2_r
\end{split}
\end{align}
Consequently, on the manifold where $\tilde{X}_2 = 0$ we obtain
\begin{align} \label{COG-eq:dlyaplem32}
\begin{split}
\dot{V}_3 \leq& \left| \tfrac{1}{C_r} \right| \left[ X_{\max} \left[ 2 \kappa_{\max} +  \tfrac{8}{\Delta} \right] -  Y_{\min} \right]v^2_r \\ &+ X(u_{rd})(F_1(\tilde{X}_1,0,u_{rd} , \dot{u}_{rd} , V_{T} , V_{N}, \Delta )\left|v_r\right| \\ &+ F_0(\tilde{X}_1,0,u_{rd} , \dot{u}_{rd} , V_{T} , V_{N}, \Delta ))
\end{split}
\end{align}
To have boundedness of $v_r$ for small values of $\tilde{X}_2$ we have to satisfy the following inequality:
\begin{align} \label{COG-eq:condlem3}
 X_{\max} \left[ 2\kappa_{\max} +  \tfrac{8}{\Delta} \right] -  Y_{\min} < 0
\end{align}
such that the quadratic term in \eqref{COG-eq:dlyaplem32} is negative. Using \eqref{COG-eq:Delta} we need to choose $\mu$, such that:
\begin{align} \label{COG-eq:mu2}
 \mu  > & \tfrac{ 8 X_{\max} }{  Y_{\min} - 2\kappa_{\max} X_{\max} }
\end{align}
which is the condition given in Lemma \ref{COG-lem3}. Note that the denominator of $\mu$ is nonzero and positive as long of the conditions of Lemma \ref{COG-lem2} are satisfied. Consequently, near the manifold $\tilde{X}_2 = 0$ it holds that \eqref{COG-eq:dlyaplem32} is negative definite for sufficiently large $v_r$.  Consequently, near the manifold $\tilde{X}_2 = 0$ it holds that \eqref{COG-eq:dlyaplem3} is negative definite for sufficiently large $v_r$. If $\dot{V}_3$ is negative for sufficiently large $v_r$ this implies that $V_3$ decreases for sufficiently large $v_r$. Since $V_3 = 1/2v^2_r$, a decrease in $V_3$ implies a decrease in $v^2_r$ and by extension in $v_r$. Consequently, $v_r$ cannot increase above a certain value and $v_r$ is bounded near $\tilde{X}_2 = 0$ if $\mu$ is chosen such that \eqref{COG-eq:mu} holds, which completes the proof of Lemma \ref{COG-lem3}.
\hfill $\blacksquare$

\end{document}